\documentclass[11pt, letter]{article}

\usepackage[margin=1in]{geometry}
\usepackage{framed}
\usepackage{microtype}
\usepackage{enumerate}

\usepackage{relsize}
\usepackage{graphicx} 
\usepackage{amsmath,amsfonts,amssymb,amsthm,amsopn}
\usepackage{tikz}
\usepackage{nicefrac}
\usepackage{paralist}
\usepackage[ruled,vlined,linesnumbered]{algorithm2e}
\providecommand{\DontPrintSemicolon}{\dontprintsemicolon}
\usetikzlibrary{shapes, arrows}

\usepackage{xcolor}
\usepackage[singlelinecheck=off,justification=raggedright]{caption}
\usepackage{floatrow}

\newfloatcommand{capbtabbox}{table}[][\FBwidth]

\usepackage{mathtools}
\usepackage[hidelinks]{hyperref}
\usepackage{color}
\theoremstyle{plain}

\newtheorem{theorem}{Theorem}[section]
\newtheorem{corollary}[theorem]{Corollary}

\newtheorem{lemma}[theorem]{Lemma}

\newtheorem{definition}[theorem]{Definition}
\newtheorem{observation}[theorem]{Observation}

\newcommand{\ProblemFormat}[1]{{\sc #1}}
\newcommand{\ProblemName}[1]{\ProblemFormat{#1}\xspace}

\newcommand{\probkmedian}[0]{\ProblemName{$k$-Median}}

\newcommand{\probeuclideankmeans}[0]{\ProblemName{Euclidean $k$-Means}}
\newcommand{\probkcenter}[0]{\ProblemName{$k$-Center}}

\newcommand{\probkmeans}[0]{\ProblemName{$k$-Means}}

\newcommand{\proborderedkmedian}[0]{\ProblemName{Ordered $k$-Median}}
\newcommand{\probmetricorderedkmedian}[0]{\ProblemName{Metric Ordered $k$-Median}}
\newcommand{\probonerectangleorderedkowa}[0]{\ProblemName{Rectangular Ordered $k$-Median}}

\newcommand{\probkcentdian}[0]{\ProblemName{$k$-Centdian}}
\newcommand{\probkfacilitypcentrum}[0]{\ProblemName{$k$-Facility $p$-Centrum}}

\newcommand{\cost}{{\mathrm{cost}}}

\newcommand{\poly}{{\mathrm{poly}}}
\newcommand{\OPT}{{\mathrm{OPT}}}
\newcommand{\OPTLP}{{\mathrm{OPT}}^{*}}
\newcommand{\ALG}{{\mathrm{ALG}}}

\newcommand{\lpc}[1]{LP(#1)\xspace}

\newcommand{\cred}{c^{\mathrm{r}}}
\newcommand{\cav}{c_{\mathrm{av}}}
\newcommand{\cavr}{\cred_{\mathrm{av}}}
\newcommand{\cavt}[1]{c^{#1}_{\mathrm{av}}}
\newcommand{\cfarr}{c^T_{\mathrm{far}}}
\DeclareMathOperator{\vol}{vol}
\newcommand{\lr}{{\ell_r}}
\newcommand{\cavlr}{{\cav^{T_r}}}
\newcommand{\wav}{w_{\mathrm{av}}}
\newcommand{\wg}{w_{\mathrm{gs}}}

\newcommand{\expected}[1]{\mathbb{E}[#1]}
\newcommand{\expectedBigPar}[1]{\mathbb{E}\left[#1\right]}

\newcommand{\Sum}{\ensuremath{\mathlarger{\sum}}}

\newcommand{\Oh}{\ensuremath{\mathcal{O}}}

\newcommand{\np}{{\mathsf{NP}}}
\newcommand{\p}{{\mathsf{P}}}
\newcommand{\calA}{\mathcal{A}}

\newcommand{\calW}{\mathcal{W}}
\newcommand{\calM}{\mathcal{M}}

\newcommand{\calB}{\mathcal{B}}
\newcommand{\calF}{\mathcal{F}}
\newcommand{\calC}{\mathcal{C}}
\newcommand{\calU}{\mathcal{U}}

\newcommand{\calI}{\mathcal{I}}

\newcommand{\reals}{{{\mathbb{R}}}}

\allowdisplaybreaks
\begin{document}

\title{Constant-Factor Approximation for \proborderedkmedian}

\date{}

\author{
Jaros{\l}aw Byrka\thanks{University of Wroc{\l}aw, Wroc{\l}aw, Poland,
\texttt{jby@cs.uni.wroc.pl}.}
\and
 Krzysztof Sornat\thanks{University of Wroc{\l}aw, Wroc{\l}aw, Poland,
\texttt{krzysztof.sornat@cs.uni.wroc.pl}.}
\and
Joachim Spoerhase\thanks{University of Wroc{\l}aw, Wroc{\l}aw, Poland,
\texttt{joachim.spoerhase@uni-wuerzburg.de}.}
}

\maketitle

\begin{abstract}
  We study the \proborderedkmedian problem, in which the solution is
  evaluated by first sorting the client connection costs and then
  multiplying them with a predefined non-increasing weight vector
  (higher connection costs are taken with larger weights). Since the
  1990s, this problem has been studied extensively in the discrete
  optimization and operations research communities and has emerged as
  a framework unifying many fundamental clustering and location
  problems such as \probkmedian and \probkcenter. This generality,
  however, renders the problem intriguing from the algorithmic
  perspective and obtaining non-trivial approximation algorithms was
  an open problem even for simple topologies such as trees.  Recently,
  Aouad and Segev were able to obtain an $\Oh(\log n)$ approximation
  algorithm for \proborderedkmedian using a sophisticated local-search
  approach and the concept of surrogate models thereby extending the
  result by Tamir (2001) for the case of a rectangular weight vector,
  also known as \probkfacilitypcentrum.  The existence of a
  constant-factor approximation algorithm, however, remained open even
  for the rectangular case.

  In this paper, we provide an LP-rounding constant-factor
  approximation algorithm for the \proborderedkmedian
  problem.

  We obtain this result by revealing an interesting connection to the
  classic \probkmedian problem.  We first provide a new analysis of
  the rounding process by Charikar and Li (2012) for \probkmedian,
  when applied to a fractional solution obtained from solving an LP
  relaxation over a non-metric, truncated cost vector, resulting in an
  elegant $15$-approximation for the rectangular case. In our
  analysis, the connection cost of a single client is partly charged
  to a deterministic budget related to a combinatorial bound based on
  guessing, and partly to a budget whose expected value is bounded
  with respect to the fractional LP-solution. This approach allows us
  to limit the problematic effect of the variance of individual client
  connection costs on the value of the ranking-based objective
  function of \proborderedkmedian.  Next, we analyze
  objective-oblivious clustering, which allows us to handle multiple
  rectangles in the weight vector and obtain a constant-factor
  approximation for the case of $\Oh(1)$ rectangles.  Then, we show
  that a simple weight bucketing can be applied to the general case
  resulting in $\Oh(\log n)$ rectangles and hence in a constant-factor
  approximation in quasi-polynomial time.  Finally, with a more
  involved argument, we show that also the clever distance bucketing
  by Aouad and Segev can be combined with the objective-oblivious
  version of our LP-rounding for the rectangular case, and that it
  results in a true, polynomial time, constant-factor approximation
  algorithm.
\end{abstract}

\thispagestyle{empty}

\newpage

\setcounter{page}{1}
\section{Introduction}

Clustering a given set of objects into $k$ groups that display a
certain internal proximity is a profound combinatorial optimization
setting. In a typical setup, we represent the objects as points in a
metric space and evaluate the quality of the clustering by a certain
function of distances within clusters. If clusters have centers and
the objective is to minimize the total distance from objects to their
cluster centers, we call the resulting optimization problem
\probkmedian. If the objective is to minimize the maximal distance to
a cluster center, then we talk about the \probkcenter problem. These
two approaches to clustering represent the two extrems in their
dependence on the variance between the indiviual connection costs in
the evaluated solution. Several intermediate approaches have been
studied such as minimizing the sum of squared connection costs known
as the \probkmeans problem.

In this paper, we study the \proborderedkmedian problem where the
connection costs are sorted non-increasingly and a
non-increasing weight vector is applied to flexibly penalize the
desired fraction of the highest costs. There is a large body of
literature on this problem because it naturally unifies many of the
most fundamental clustering and location problems such as
\probkmedian, \probkcenter, \probkcentdian, and \probkfacilitypcentrum
(see below for definitions). We refer to the book of Nickel and
Puerto~\cite{Nickel09-ordered-k-median-book} dedicated to
\proborderedkmedian problems for an extensive overview. See also below
for a selection of related works and also applications in
multi-objective optimization and robust optimization.

The generality of \proborderedkmedian renders it intriguing
from the computational perspective~\cite{AouadS17:logn-apx}. For
example, whereas \probkmedian and \probkcenter can be solved
efficiently on trees by dynamic programming such approaches seem to
fail for \proborderedkmedian due to the lack of separability
properties~\cite{PuertoT05:treeshaped-ordered}. Regarding
approximability in general metric spaces, constant-factor
approximation algorithms are long known for
\probkmedian~\cite{CharikarGTS99-kmedian} and
\probkcenter~\cite{HochbaumS85}. In contrast, not even non-trivial
\emph{super-constant} approximability results were known for
\proborderedkmedian until very recently and even developing
constant-factor approximations for seemingly simple topologies such as
trees turned out non-trivial~\cite{AouadS17:logn-apx}. In particular,
due to the non-linearity of the objective function there seems to be
no obvious way to apply tools such as metric tree
embeddings~\cite{AouadS17:logn-apx,bartal98:tree-embeddings}. To
demonstrate the highly non-local and dependent structure of the
objective function note that even if the clusters are given the
selection of the cluster centers cannot be made solely on a
per-cluster basis but depends on the decision in other clusters due to
the ranking of distances in the objective.

\paragraph{Related works for \proborderedkmedian.} The problem
generalizes many fundamental clustering and location problems
\cite{CharikarGTS99-kmedian,jain2003greedy,arya2004local,li2016approximating,ByrkaPRST17,HochbaumS85,tamir98:centdiantree,Tamir01}
such as the above-mentioned \probkmedian, \probkcenter problems, the
\probkcentdian problem where the objective is a convex combination of
the \probkmedian and the \probkcenter objective, or the
\probkfacilitypcentrum problem where the objective accounts for the
$p$ highest connection costs. The ordered median objective function
has also been considered in robust
optimization~\cite{PuertoRT17-k-sum,BertsimasS03-robust-k-sum,bertsimas05:opt-over-integers,bertsimas2014-least-quantile-regression}
and multi-objective
optimization~\cite{FernandezPPS17-OWA-multi-obj-ST}. For a
comprehensive overview we refer to the
books~\cite{Nickel09-ordered-k-median-book,laporte-etal:location-science} on \proborderedkmedian
problems and to dedicated works~\cite{NickelP99-location-problems-unified,Nickel:ordered-weber,BradleyFM99-math-prog-form-chall,puerto-fernandez:geometrical-props-single-facility}.

Due to the above-outlined difficulties in obtaining algorithmic
results for the general problem, structural properties of continuous
network spaces have been
studied~\cite{PuertoR05-exp-card-ordered-median,DreznerN09-dc-decomposition-median,EspejoRV09-convex-ordered-median,Rodríguez-Chía2005:ordered-networks}
where facilities may be placed at interior points on edges,
single-facility
models~\cite{NickelP99-location-problems-unified,DreznerN09-ordered-one-median-plane,KalcsicsNPT02-algorithic-ordered,EspejoRV09-convex-ordered-median},
and multi-facility models on special topologies such as trees
\cite{KalcsicsNP03-multifacility-ordered,Tamir01,PuertoRT17-k-sum}. Also,
integer programming formulations~\cite{Nickel:ordered-weber},
branch-and-bound
methods~\cite{BolandDNP06:exact-ordered-median,PuertoRR13-branch-cut-ordered-median},
heuristics
\cite{Dominguez-MarinNHM05-heuristic,StanimirovicKD07-genetic-ordered,PuertoPG14-modified-vns,labbe17-comparative-study-ordered-median},
and related location models~\cite{puerto2009minimax,PuertoT05:treeshaped-ordered} have been studied.
For a survey on the topic we refer to~\cite{Nickel09-ordered-k-median-book}.

Although the problem has received much attention in the discrete
optimization and operations research literature, obtaining any
non-trivial approximation algorithm for \proborderedkmedian had been
an intriguing open question until very recently, when Aouad and Segev
\cite{AouadS17:logn-apx} were able to devise an $\Oh(\log n)$
approximation algorithm for the problem using a sophisticated
local-search approach and the concept of \emph{surrogate
  models}.

Below, we list approximability results for certain specific objective
functions that fall under the framework of \proborderedkmedian or that
are closely related and where approximability has been studied for
general metrics.

\paragraph{Approximation algorithms for \probkmedian, \probkcenter, and \probkmeans.}
\probkmedian admits constant-factor approximations via local-search~\cite{arya2004local},
or a direct rounding of the standards LP~\cite{CharikarL12:k-median}.
The current best ratio of $(2.675+\epsilon)$~\cite{ByrkaPRST17} is obtained
by combining a primal-dual algorithm~\cite{jain2003greedy}, and a nontrivial rounding of a so-called bi-point solution based on a preprocessing introduced in~\cite{li2016approximating}.

The situation with the \probkcenter setting is simpler. A simple 2-approximation\footnote{In the setting where cluster centers are selected from a different set the method of~\cite{HochbaumS85} gives a tight 3-approximation.} is obtained
via guessing the longest connection distance in the optimal solution~\cite{HochbaumS85}, 
and this is tight assuming $\p \neq \np$. Notably, by contrast to the \probkmedian setting,
the most natural LP for \probkcenter has unbounded integrality gap.

Also \probkmeans admits a constant-factor approximation. The
$(9+\epsilon)$-approximation algorithm for
\probeuclideankmeans~\cite{KanungoMNPSW04} can be shown to provide
$25$-approximation in general metrics. The recent work of Ahmadian et
al.~\cite{AhmadianNSW16} decreases these ratios to $6.375$ and $9$,
respectively.

\paragraph{Approximation algorithms for further specific objective
  functions.}
A special case of \proborderedkmedian that we call
\probonerectangleorderedkowa was considered by Tamir~\cite{Tamir01}
(who called it \probkfacilitypcentrum). In this setting, we have to
open exactly $k$ facilities and the objective function is just a sum
of $p$ largest client connection costs. He gives polynomial time
algorithms that solve the problem (optimally) on path and tree graphs
and using tree-embedding gives a $\Oh(\log n)$-approximation for this
case. Obtaining a constant-factor approximation for
\probonerectangleorderedkowa, however, has been an open
problem~\cite{Tamir01,AouadS17:logn-apx}.

One should also notice at least two further, very recent works combining the
\probkcenter objective and the \probkmedian objective. First, Alamdari
and Shmoys~\cite{AlamdariS:waoa17} considered a bicriteria
approximation algorithm for the \probkcenter and \probkmedian
problems, i.e., the objective function is a linear combination of two
objectives that are maximal connection cost in use and sum of all used
connection costs. This problem is known as \probkcentdian
\cite{tamir98:centdiantree}.  They obtained polynomial time bicriteria
approximation of $(4,8)$, where the first factor is in respect to the
\probkcenter objective and the second factor is in respect to the
\probkmedian objective. (Alamdari and Shmoys note, however, that the
two problems \probkmedian and \probkcenter are \emph{not} approximable
simultaneously.)  Also \probkcentdian is a special case of
\proborderedkmedian. The second recent work combining the \probkmedian and
the \probkcenter objective is the work of Haris et
al.~\cite{harris2017symmetric} who propose a method to select $k$
facilities that deterministically guarantees each client to have a
connection within a certain fixed radius but also provides a stronger
per client bound on cost expectation.

\paragraph*{Relation to the work of Chakrabarty and
  Swamy~\cite{ChakrabartyS17}.} Soon after the submission of our
paper, Chakrabarty and Swamy~\cite{ChakrabartyS17} announced
constant-factor approximation algorithms for
\probonerectangleorderedkowa and also for (general)
\proborderedkmedian. The part of their argument for
\probonerectangleorderedkowa appears to be obtained
independently. Instead of the LP-rounding process of Charikar and
Li~\cite{CharikarL12:k-median}, they either use a primal-dual approach
or a black-box reduction to \probkmedian.


\subsection{Our Results and Techniques}
Our main result is an LP-rounding constant-factor approximation
algorithm for the \proborderedkmedian problem.

We are not aware of a LP relaxation for \proborderedkmedian with
bounded integrality gap. In our approach we guess a \emph{reduced cost
function} roughly mimicking the weighting of distances in an optimum
solution and solve the natural LP relaxation for \probkmedian under
this reduced cost function (rather than under the original
metric). Subsequently, we round this solution via a dependent LP
rounding process by Charikar and Li \cite{CharikarL12:k-median} for
\probkmedian operating on the original (unweighted) metrics.

The challenge and our main technial contribution consists in analyzing
the approximation performance of this approach. In the original
analysis of Charikar and Li~\cite{CharikarL12:k-median} for the
\probkmedian objective, a per-client bound on the expected connection
cost of this client with respect to its fractional connection cost is
established.  The global approximation ratio is then obtained by
linearity of expectation.  The above-described non-linear,
ranking-based character of the objective of \proborderedkmedian poses
an obstacle to apply an analogous reasoning also in our more general
setting as the actual weight that is applied to the connection cost of
a client depends highly on the (random) connection costs of the other
clients.

We use four key ingredients to overcome this technical hurdle.

First, we show that the algorithm provides a constant-factor
approximation for rectangular weight vectors. This already answers the
open problem stated in~\cite{AlamdariS:waoa17,AouadS17:logn-apx}.  In
our analysis, the connection cost of a single client is partly charged
to a \emph{deterministic budget} related to a \emph{combinatorial bound}
based on guessing, and partly to a \emph{probabilistic 
budget} whose expected value is bounded with respect to the 
fractional LP-solution. This approach allows to limit the
above-described problematic effect of the variance of individual
client connection costs on the value of the ordered objective function
of \proborderedkmedian.

Second, we show a surprising \emph{modularity} of Charikar and Li's
rounding process. The solution computed by this process can be related
to the above-mentioned combinatorial and fractional bounds
simultaneously with respect to \emph{all} rectangular objectives. This
property is oblivious to the objective with respect to which the input
fractional solution was optimized.

Third, we \emph{decompose} an arbitrary non-increasing weight vector
into a convex combination of rectangular objectives. The
aforementioned modularity property provides a bound for each of those
objectives.  We show that those bounds nicely combine to a global
bound on the approximation ratio giving a constant-factor
approximation with respect to a combinatorial bound and a fractional
bound both under the original, general weight objective.

A straightforward application of this approach incorporating weight
bucketing gives only quasi-polynomial time due to the guessing part.
To achieve a truly polynomial time algorithm we apply a clever
distance bucketing approach by Aouad and Segev
\cite{AouadS17:logn-apx}, which guesses for each distance bucket the
average weight applied to this bucket by some optimal solution. Our
analysis approach applies also to this more intricate setting but
turns out technically more involved.

\section{Definitions}\label{sec:definitions}

\begin{definition}
 In the \probmetricorderedkmedian problem we are given: a finite set of facilities $\calF$,
 a set of clients $\calC$, $|\calC| = n$, a metric cost function $c\colon\calF \times \calC \to \reals_{\geq 0}$,
 an integer $k \geq 1$ as the number of facilities to open and a non-increasing weight vector $w = (w_1, \dots, w_n)$.
 For a subset $\calW \subseteq \calF$ and client $j\in\calC$, we define
 $c_j(\calW) = \min_{i \in \calW} c(i,j)$ as the smallest {\it connection cost of $j$} to a facility in $\calW$.
 We sort the values $c_j(\calW), j \in \calC$ in non-increasing order i.e. we define
 $c^\to(\calW) = (c_j^\to(\calW): 1 \leq j \leq n)$ such that
 $\{c_j^\to(\calW): 1 \leq j \leq n\} = \{c_j(\calW): j \in \calC \}$ and
 $c_j^\to(\calW) \geq c_{j'}^\to(\calW)$ whenever $1 \leq j \leq j' \leq n$.
 {\it The connection cost} of $\calW$ is the weighted sum $\cost(\calW) = \sum_{j = 1}^n w_j \cdot c_j^\to(\calW)$.
 The goal is to find a set $\calW \subseteq \calF, |\calW|=k$ that minimizes the connection cost.
\end{definition}

In the rest of the paper we say \proborderedkmedian for
\probmetricorderedkmedian because non-metric cost function does not
allow us to obtain any non-trivial approximation (unless $\p=\np$).  In what follows we will assume w.l.o.g.\ that $w_1=1$ in the above definition. Let $j\in\calC$ be a client.  Then $\calB(j,r)$ denotes the set of all facilities $i\in\calF$ with $c_{ij}<r$, that is, $\calB(j,r)$ is an open ball (in the set of facilities) of radius $r$ around $j$.

\begin{definition}
  Consider an instance of \proborderedkmedian.  A \emph{reduced cost function} $\cred$ is a (not necessarily metric) function $\cred\colon\calF \times \calC \to \reals_{\geq 0}$ such that for all $i,i'\in\calF$ and $j,j'\in\calC$ we have that $\cred_{ij}\leq c_{ij}$ and that $c_{ij}\leq c_{i'j'}$ implies $\cred_{ij}\leq \cred_{i'j'}$. 
\end{definition}
Reduced cost functions arise naturally for \proborderedkmedian since in its objective function non-increasingly sorted distances are multiplied by non-increasing weights $\leq 1$.

\begin{definition}
 \probonerectangleorderedkowa is a special case of \proborderedkmedian problem with weights $w_1 = w_2 = 
\dots = w_\ell = 1$  and $w_{\ell+1} = w_{\ell+2} = \dots = w_{n} = 0$ for some 
$\ell \in [n]$. For any $\calW\subseteq\calF$ let $\cost_\ell(\calW)$ denote the objective function of $\calW$ for this problem.
\end{definition}

Note that \probonerectangleorderedkowa with $\ell=1$ is equivalent to \probkcenter
and \probonerectangleorderedkowa with $\ell=n$ is equivalent to \probkmedian. In what follows, missing proofs can be found in the appendix.

\section{Algorithmic Framework}\label{sec:algor-fram}

Our algorithms consist of two parts: An \emph{LP-solving} and an
\emph{LP-rounding} part.

In the \emph{LP-solving} part, we compute an optimal solution to an
LP-relaxation, which is (apart from the objective function) identical
to the standard LP relaxation for \probkmedian.  However, instead of
using the input metrics $c$ in the objective function, we employ a
reduced cost function~$\cred$. Intuitively in ~$\cred$ the distances are multiplied by
roughly the same weights as in a guessed optimal solution. 

In the \emph{LP-rounding} part the fractional solution provided by the
above-described guessing will be rounded to an integral solution by
applying the algorithm of Charikar and
Li~\cite{CharikarL12:k-median}. In contrast to the LP-solving part,
this algorithm operates, however, in the \emph{original} metric space
rather than in the (generally non-metric) reduced cost space.

\subsection{LP-Relaxation}\label{sec:lp-relaxation}

Let \lpc{$\cred$} be the following relaxation of a natural
ILP formulation of \probkmedian under some reduced cost function $\cred$.
\begin{align}
 \text{minimize} \quad \sum_{i \in \calF, j \in \calC}^{} \: \cred_{ij}x_{ij} & \quad\quad \text{s.t.}&\label{lp:objective}\\
 x_{ij} & \leq y_i \quad\quad i \in \calF, j \in \calC &\label{lp:fractional_opening}\\
 \sum_{i\in\calF} x_{ij} & = 1 \quad\quad j \in \calC &\label{lp:fully_served}\\
 \sum_{i\in\calF} y_{i} & = k &\label{lp:k_opened}\\
 0 \leq x_{ij}, y_{i} & \leq 1 \quad\quad i \in \calF, j \in \calC &\label{lp:non_negative}
\end{align}

Here, $y_i$ denotes how much facility $i$ is open (0---closed, 1---opened) and 
$x_{ij}$ indicates how much client $j$ is served by facility $i$
(0---non-served, 1---served).
Equality~\eqref{lp:k_opened} ensures that exactly $k$ facilities are
opened (possibly fractionally), \eqref{lp:fully_served} guarantee that each
client is served (possibly fractionally). \eqref{lp:fractional_opening}
do not allow a facility to serve a client more than how much it is opened.
For each client $j\in\calC$ let
$\cavr(j)=\sum_{i\in\calF}\cred_{ij}x_{ij}$ denote the
fractional (or average) reduced connection cost of $j$.

\subsection{Guessing and LP-Solving}
Note that if $\cred=c$ where $c$ is the input metrics, \lpc{$c$} becomes the standard LP relaxation for the classical \probkmedian objective.  In order
to obtain a valid lower bound \lpc{$\cred$} for a \proborderedkmedian instance,
we employ guessing of certain distances in an optimal solution.
%
%
The details of the guessing are setting-specific and are thus
described later. We can w.l.o.g.\ assume that we compute an optimal
solution $(x,y)$ to \lpc{$\cred$} such that for all
$i\in\calF, j\in\calC$ we have that $x_{ij}\in\{0,y_i\}$ and that
$y_i>0$. Also, since $\cred$ preserves the order of the distances, we
can assume that if $y$ is kept fixed then $x$ optimizes
$\sum_{i\in\calF,j\in\calC}c_{ij}x_{ij}$ and say that $x$ be
\emph{distance-optimal}. See appendix for details.

\subsection{LP-Rounding: Dependent Rounding Approach of Charikar and Li}

We round the fractional solution obtained in the LP-solving phase to an
integral solution by the (slightly modified) LP-rounding process
of Charikar and Li \cite{CharikarL12:k-median} for \probkmedian.

To apply this algorithm note that the feasibility of a solution
$(x,y)$ to \lpc{$\cred$} does \emph{not} depend on the cost vector $\cred$.  This
enables us to compute an optimum solution $(x,y)$ to \lpc{$\cred$} for some
appropriate reduced cost function and to subsequently apply the rounding
process of Charikar and Li (which operates on the original metrics
$c$) to the solution $(x,y)$.  In the analysis, we have to
exploit how $\cred$ and $c$ are related in order to bound the approximation ratio of the algorithm.

We now describe the rounding algorithm of Charikar and Li, which
consists of four phases: a clustering phase, a bundling phase, a
matching phase, and a sampling phase (see Algorithm~\ref{alg:main}).
 \begin{algorithm}[h]
  
  \SetAlgorithmName{Algorithm}{}
  \DontPrintSemicolon 
  \SetNoFillComment
  \SetKwComment{tcc}{}{}
  \small

    \KwData{feasible fractional solution $(x,y)$ to \lpc{$\cred$} satisfying the properties of Lemma~\ref{lem:canonical-solution}}
    \KwResult{set of $k$ facilities}
     

     \smallskip
     \tcc{/* Clustering phase */}
     \tcc{ \textcolor{gray}{/* run a \emph{clustering procedure} to compute a set $\calC'\subseteq\calC$ of \emph{cluster centers} so that each client $j\in\calC$ is ``close'' to some cluster center $j'\in\calC'$ and so that the cluster centers are ``far'' from each other */}}
     $\calC'\gets$ Clustering$(x,y)$\;


        \smallskip
     \tcc{/* Bundling phase */}
     \For{$j \in \calC'$}{
     $R_j \gets \frac{1}{2} \min_{j' \in \calC' ,j' \neq j} (c_{jj'})$\;
     $\calU_j \gets \{i \in \calF: i \in \calB(j,R_j)\textnormal{ and } x_{ij}>0\}$\;
     }

        \smallskip
     \tcc{/* Matching phase */}

     $\calM \gets \emptyset$\;
     \While{there are unmatched clients in $\calC'$}{
     add to $\calM$ a pair from $\calC'$ that is the closest pair among unmatched clients in $\calC'$
     }

        \smallskip
     \tcc{/* Sampling phase (dependent rounding) */}
     \tcc{\textcolor{gray}{/* Apply dependent randomized rounding as described by Charikar and Li \cite{CharikarL12:k-median} preserving the marginals for the individual facilities, bundles, matched pairs in $\calM$, and set $\calF$ */}}
     \Return DependentRounding$(x,y,\{U_j\}_j,\calM,\calF)$

  \caption{Rounding Algorithm by Charikar and Li \cite{CharikarL12:k-median}}\label{alg:main}
 \end{algorithm}
 Below we give some intuition on the different phases. More formal
 arguments will be given later.

 The purpose of the \emph{clustering procedure} is to compute a set
 $\calC'\subseteq\calC$ of \emph{cluster centers} so that each client
 $j\in\calC$ is ``close'' to some cluster center $j'\in\calC'$ and so
 that the cluster centers are ``far'' from each other. We thus may
 think of the cluster centers representing all remaining clients. The
 implementation of the procedure and the meaning of ``close'' and
 ``far'' is application-specific and will thus be described later.

 In the \emph{bundling phase} each cluster center $j\in\calC'$ is
 associated with a bundle $\calU_j$ of facilities of total fractional opening at least $1/2$, so that we can show
 that the bundles are pairwise disjoint\footnote{Our version of the algorithm is actually slightly simpler here than the original one, which is sufficient for constant-factor approximations.  In the original versions the bundles are based on larger balls that may overlap but are made disjoint in a greedy manner.  This allows Charikar and Li to obtain an improved approximation factor.}.
 
 In the \emph{matching phase} cluster centers are paired in a greedy
 manner.  As we will show that the volume of each bundle is at least
 $1/2$ the total volume of the bundles of a matched pair is at least
 $1$.  This will ensure that in the subsequent \emph{sampling phase}
 at least one facility per pair is opened.

 In the sampling phase we use the dependent randomized rounding
 procedure described by Charikar and Li \cite{CharikarL12:k-median} to
 open facilities and obtain a feasible solution. We do
 not describe the details of the implementation here but use it as
 a ``black box'' satisfying the following properties (as in the original work of Charikar and Li):

 \begin{lemma}\label{lem:lp-round-depend}
 The procedure DependentRounding in Algorithm~\ref{alg:main} can be implemented such that the following holds for any input $(x,y)$ being a feasible solution to \lpc{$\cred$}.
 \begin{enumerate}[(i)]
 \item\label{item:prob-facility} Each facility $i\in\calF$ is opened with probability precisely $y_i$,
   \item\label{item:prob-bundle} in each bundle $\calU_j$ with $j\in\calC'$ a facility is opened with probability precisely $\vol(\calU_j)$,
   \item\label{item:prob-pair} for each matched pair $(j,j')$ in $\calM$ at least one
     facility in $\calU_j\cup\calU_{j'}$ will be opened,
   \item\label{item:total-cardinality} in total at most $k$ facilities are opened.
   \end{enumerate}
 \end{lemma}

\section{Rectangular Weight Vectors}\label{sec:one_rectangle_apx}



\begin{theorem}\label{thm:alg_for_one_rectangle}
There exists a polynomial time randomized algorithm for \probonerectangleorderedkowa
that gives a $15$-approximation in expectation.
\end{theorem}

To proof this theorem, we need to fill in the
following two missing parts of the framework: Guessing of the reduced
cost space and the clustering procedure in the rounding part.

\subsection{Guessing and Reduced Costs}

In the LP-solving phase, we guess the value $T$
of $\ell$-th largest
distance in an optimum solution to \probonerectangleorderedkowa.  (This
is the smallest distance that is counted in the total connection cost
with non-zero weight.) As the correct guess of $T$ is the distance
between a client and a facility the guessing can be performed by
considering only $\Oh(mn)$ options for $T$.

For each $i\in\calF, j\in\calC$, we define the
\emph{reduced cost}
\begin{equation}\label{def-cijT}
 c^T_{ij} = 
 \begin{cases}
 c_{ij} \quad &\text{if} \quad c_{ij} \geq T,\\
 0 \quad &\text{otherwise},
 \end{cases}
\end{equation}
that will be used as a cost function in our LP for the \proborderedkmedian.

An optimal solution $(x,y)$ to \lpc{$c^T$} is a feasible solution for
\lpc{$c$} as well. As introduced in Section~\ref{sec:lp-relaxation},
we use $\cav(j) = \sum_{i \in \calF} x_{ij} \cdot c_{ij}$ and
$\cavt{T}(j) = \sum_{i \in \calF} x_{ij} \cdot c^T_{ij}$ to denote the
average connection cost and the average reduced connection cost of a
client $j\in\calC$, respectively.


\subsection{Dedicated Clustering}
The following two clustering methods will be considered. We first analyze using Algorithm~\ref{alg:dedicated-clustering}.

\noindent
\begin{minipage}[h]{0.5\textwidth}
\null 
        \begin{algorithm}[H]
            \SetAlgorithmName{Algorithm}{}
            \DontPrintSemicolon 
            \SetNoFillComment
            \SetKwComment{tcc}{}{}
            \small

            \KwData{feasible fractional solution $(x,y)$ to \lpc{$c$}}
            \KwResult{set $\calC'\subseteq\calC$ of cluster centers}

            $\calC' \gets \emptyset$\;
            $\calC'' \gets \calC$\;
            $\underline{\cavt{T}(j)} \gets \sum_{i \in \calF} x_{ij} \cdot \underline{c_{ij}^{T}}$ for all $j\in\calC$; \hfill {\bf DIFF}

            \While{$\calC''$ is non empty}{
                take $j \in \calC''$ with the smallest \underline{$\cavt{T}(j)$}; \hfill{\bf DIFF}

                add $j$ to $\calC'$\;
                delete from $\calC''$ client $j$\;
                delete from $\calC''$ all clients $j'$ with $c_{jj'} \leq \underline{4\cavt{T}(j') + 4T}$ \hfill {\bf DIFF}\label{alg:line:plus4T}
            }

            \Return $\calC'$
            \caption{DedicatedClustering}\label{alg:dedicated-clustering}
        \end{algorithm}
\end{minipage}%
\begin{minipage}[h]{0.5\textwidth}
\null
        \begin{algorithm}[H]
            \SetAlgorithmName{Algorithm}{}
            \DontPrintSemicolon 
            \SetNoFillComment
            \SetKwComment{tcc}{}{}
            \small

            \KwData{feasible fractional solution $(x,y)$ to \lpc{$c$}}
            \KwResult{set $\calC'\subseteq\calC$ of cluster centers}

            $\calC' \gets \emptyset$\;
            $\calC'' \gets \calC$\;
            $\underline{\cav^{\phantom{T}}(j)} \gets \sum_{i \in \calF} x_{ij} \cdot \underline{c_{ij}^{\phantom{T}}}$ for all $j\in\calC$\;
            \While{$\calC''$ is non empty}{
                take $j \in \calC''$ with the smallest \underline{$\cav^{\phantom{T}}(j)$}\;
                add $j$ to $\calC'$\;
                delete from $\calC''$ client $j$\;
                delete from $\calC''$ all clients $j'$ with $c_{jj'} \leq \underline{4\cav^{\phantom{T}}(j')}\phantom{ + 4T}$ \label{alg:line:clustering}
            }

            \Return $\calC'$
            \caption{ObliviousClustering}\label{alg:oblivious-clustering}
        \end{algorithm}
\end{minipage}

This clustering procedure is very similar to the one of
Charikar and Li (see also Section~\ref{sec:oblivious-clustering} below) except for the fact that the procedure needs to know the threshold $T$ of the guessing phase. (Note that we use $T$ explicitly but also implicitly in the
average reduced cost $\cavt{T}(j)$.) This dependence allows a simpler and
better analysis for \probonerectangleorderedkowa. In Section~\ref{sec:oblivious-clustering}, we will describe how to get
rid of this dependency, which allows us to generalize the result.

\subsection{Analysis of the Algorithm}
In the following we analyze Algorithm~\ref{alg:main} using the procedure DedicatedClustering.

%

\begin{observation}
  We have $c^T_{ij}\leq c_{ij}\leq c^T_{ij}+T$ for any $i \in \calF, j \in \calC$
  and consequently $\cavt{T}(j) \leq \cav(j)\leq \cavt{T}(j)+T$.
\end{observation}

\begin{lemma}\label{lem:filtering}
  The following two statements are true
  \begin{enumerate}[(i)]
  \item\label{item:filtering-spread} For any $j,j'\in\calC'$ we have that $c_{jj'}>4\max(\cavt{T}(j),\cavt{T}(j'))+4T$.
  \item\label{item:filtering-assign} For any $j\in\calC\setminus\calC'$
   there is a client $j'\in\calC'$ with $\cavt{T}(j')\leq\cavt{T}(j)$ and $c_{jj'}\leq 4\cavt{T}(j)+4T$.
  \end{enumerate}
\end{lemma}

\begin{lemma}\label{lem:bundles}
  The following two statements are true
  \begin{enumerate}[(i)]
  \item\label{item:bundles-volume} $\vol(\calU_j) \geq 0.5$ for all $j\in\calC'$
  \item\label{item:bundles-disjoint} $\calU_j\cap\calU_{j'}=\emptyset$ for all
   $j,j'\in\calC'$, $j\neq j'$
  \end{enumerate}
\end{lemma}

\begin{proof}[Proof of Theorem~\ref{thm:alg_for_one_rectangle}.]
  Let $\OPT$, $\OPTLP$ and $\ALG$ be the values of an optimal solution,
  \lpc{$c_r$} and Algorithm~\ref{alg:main}, respectively.
  Note that $\OPTLP \leq \OPT$.
  
  The idea of the proof is to provide for each client an upper bound
  on the distance $C_j$ (according to the original distance $c$) traveled by
  this client. The upper bound is paid for by two budgets $D_j$ and $X_j$. The
  ``deterministic'' budget $D_j$ is $5T$. The ``probabilistic''
  budget $X_j$ is a random variable (depending on the random choices
  made by the algorithm).

  We will show below that (by a suitable choice of $X_j$) the
  connection cost $C_j$ of $j$ can actually be upper bounded by
  $D_j+X_j$ and that $\expected{X_j} \leq 10 \cdot \cavt{T}(j)$.  We
  claim that this will complete our proof of a 15-approximation.  To
  this end, note that at most $\ell$ clients $j$ will pay the
  deterministic budget $D_j=5T$ because at most $\ell$ distances are
  actually accounted for in the objective function. Unfortunately, an
  analogous reasoning does not hold true for the expected value of the
  random variables $X_j$.
  (For example, note that $\expected{\max (X_1,\dots,X_n)}$ is
  generally unbounded in $\max(\expected{X_1},\dots,\expected{X_n})$
  in the case of $\ell=1$.)  However, we can just sum over \emph{all}
  those random variables obtaining an upper bound on the total
  expected connection cost:
\begin{displaymath}
  \expected{\ALG} \leq D_j \cdot \ell + \sum_{j \in C} \expected{X_j} \leq
  5\ell\cdot T+10 \cdot \sum_{j\in C}\cavt{T}(j)\leq 15 \cdot \OPT.
\end{displaymath}
For the last inequality note that by our guess of $T$ we have that
$\textnormal{OPT}\geq\ell\cdot T$ and from the~definition of \lpc{$c_r$}
we have $\OPTLP = \sum_{j\in C}\cavt{T}(j)$.
To show that $C_j \leq D_j + X_j$ consider an arbitrary client $j$ with connection
cost $C_j$. We incrementally construct our upper bound on $C_j$
starting with $0$.  Each increment will be either charged to $D_j$ or
$X_j$.

Consider a client $j$ and the cluster center $j'$ it is assigned to
(possibly $j=j'$).  We have that $c_{jj'}\leq 4\cavt{T}(j)+4T$ by
Lemma~\ref{lem:filtering}~(\ref{item:filtering-assign}). We charge
$4T$ to $D_j$ and $4\cavt{T}(j)$ with probability~1 to $X_j$.

We now describe how to pay for the transport from $j'$ to an open
facility. There are two cases to distinguish.
Either a facility
within a radius $T$ around $j'$ is opened or not. If yes, then this cost can be
covered by charging an additional amount of $T$ to $D_j$.
In this case the total cost is upper bounded by $D_j = 5T$ and
$\expected{X_j} = 4\cavt{T}(j)$.

If no facility within a radius $T$ around $j'$ is opened then observe that for
all facility $i$ with $c_{ij'}\geq T$ we have that
$c^T_{ij'}=c_{ij'}$.  We now continue to bound the connection cost for
this case.  Let $j''$ be the closest client distinct from $j'$ in
$\calC'$. We consider the case where $j'$ and $j''$ are not
matched. (The case where they are matched is simpler.) Let $j'''$ be
the client in $\calC'$ to which $j''$ is matched, i.e., $(j'',j''') \in \calM$.
By the dependent
rounding process one facility in $\calU_{j''}\cup\calU_{j'''}$ will be
opened.  We have that $c_{j'j''}=2R_{j'}=2R$ and thus $c_{j''j'''}\leq 2R$ and $R_{j''},R_{j'''}\leq R$
(otherwise, $j''$ would not have been matched with $j'''$ but with $j'$).

This means that, in case no facility is opened in the bundle $\calU_{j'}$
the client $j$ travels an additional distance of at most
$\max(c_{j'j''} + R_{j''}, c_{j'j''} + c_{j''j'''} + R_{j'''}) \leq 2R + 2R + R\leq 5R$.

If a facility is opened in the bundle $\calU_{j'}$ then we charge this
additional connection cost to $X_j$. The contribution of this case to $\expected{X_j}$ is (by Properties~(\ref{item:prob-facility}) and~(\ref{item:prob-bundle}) of Lemma~\ref{lem:lp-round-depend})
at most 
\begin{displaymath}
  \sum_{i \in \calU_{j'} \setminus \calB(j',T)} \hspace{-15pt} y_ic_{ij'} =
  \sum_{i \in \calU_{j'} \setminus \calB(j',T)} \hspace{-15pt} x_{ij'}c_{ij'} =
  \sum_{i \in \calU_{j'} \setminus \calB(j',T)} \hspace{-15pt} x_{ij'}c^T_{ij'} \leq
  \sum_{i \in \calF_{j'}} x_{ij'}c^T_{ij'} = \cavt{T}(j')\,.
\end{displaymath}
Here, the first equality follows by Property~(\ref{item:client-fully-assigned}) of Lemma~\ref{lem:canonical-solution}.
The second equality follows because we assume that no facility is opened in $\calB(j',T)$ and since $c_{ij'}=c^T_{ij'}$ for all $i\in\calU_{j'}\setminus \calB(j',T)$.

We finally handle the case where no facility in $\calU_{j'}$ is opened
and where $j$ additionally travels a distance of at most $5R$.  We
charge this addtional cost to $X_j$. We bound the probability that
this case occurs. We claim that $\vol(\calU_{j'})$ is at least
$1-\cavt{T}(j')/R$. To see this, recall that
$2R \geq c_{j''j'''} > 4\max(\cavt{T}(j''),\cavt{T}(j'''))+4T$ thus $R>T$.
Note that the reason of adding the quantity $4T$ in the clustering phase
(line~\ref{alg:line:plus4T}) is to have the property $R > T$ (in the original algorithm of Charikar-Li~\cite{CharikarL12:k-median} this property is not necessarily satisfied).
Using this, for all facilities in $\calF_{j'}\setminus\calU_{j'}$
we have that $c^T_{ij'}=c_{ij}$ because $R>T$.
Hence
\begin{align*}
 \cavt{T}(j') &\geq \sum_{i \in \calF_{j'} \setminus \calB(j',T)} \hspace{-15pt} x_{ij'} \cdot c^T_{ij'} =
 \sum_{i \in \calF_{j'} \setminus \calB(j',T)} \hspace{-15pt} x_{ij'} \cdot c_{ij'} \geq
 \sum_{i \in \calF_{j'} \setminus \calU_{j'}} \hspace{-10pt} x_{ij'} \cdot c_{ij'} \\
 &\geq R \cdot \hspace{-10pt} \sum_{i \in \calF_{j'} \setminus \calU_{j'}} x_{ij'} =
 R \cdot \vol(\calF_{j'} \setminus \calU_{j'}) =
 R \cdot (1-\vol(\calU_{j'})),
\end{align*}
which implies the claim. Here, note that $\calB(j',T)  \subseteq \calU_{j'}$ because $R>T$.
This means that $j$ travels the additional distance of $5R$ with probability
at most $\cavt{T}(j')/R$ and hence the contribution to $\expected{X_j}$ is
upper bounded by $5 \cdot \cavt{T}(j')$.
Summarizing, for the case when no facility is opened within $\calB(j',T)$ we can upper bound
$\expected{X_j}$ by:
\begin{itemize}
 \item a cost of serving client $j$ through
 the closest cluster center $j'$ that is $4 \cdot \cavt{T}(j)$, plus
 \item a value $\cavt{T}(j')$ for the case when a facility is opened within bundle $\calU_{j'}$, plus
 \item a value $5R$ with probability at most $\cavt{T}(j')/R$ when no facility is opened within $\calU_{j'}$.
\end{itemize}
Hence $\expected{X_j} \leq 4 \cdot \cavt{T}(j) \cdot 1 + \cdot \cavt{T}(j') \cdot
1 + 5R \cdot \cavt{T}(j')/R \leq 10 \cdot \cavt{T}(j),$
by taking into account that $\cavt{T}(j')\leq\cavt{T}(j)$.
Moreover we charged again at most $5T$ to $D_j$ in this case.
In the end we have the desired two upper bounds for both budgets for completing the proof:
$D_j \leq 5T, \; \expected{X_j} \leq 10 \cdot \cavt{T}(j)$.
\end{proof}

\subsection{Oblivious Clustering}\label{sec:oblivious-clustering}
In Algorithm~\ref{alg:main}, we are working on the original metrics
$c$ but still DedicatedClustering described in the
previous section depends on our guessed parameter $T$ and the reduced
metrics $c^T$.  In this section, we show that we can apply the
original clustering of Charikar and Li that works solely on the
input metrics $c$ and that is thus \emph{oblivious} of the guessing phase.
In particular, we use the Oblivious Clustering procedure as described in
Algorithm~\ref{alg:oblivious-clustering}.

Using Oblivious Clustering, we can prove the following version of
Theorem~\ref{thm:alg_for_one_rectangle}. While the constants proven in
the following lemma are weaker than the ones for DedicatedClustering, it exhibits a
surprising modularity that is a key ingredient to later handle the
general case.  In particular, the clustering (and thus the
whole rounding phase) are unaware (oblivious) of the cost vector
$\bar{c}$ with respect to which we optimized
\lpc{$\bar{c}$}. Secondly, the bound proven in the lemma holds for
\emph{any} rectangular objective function of \proborderedkmedian
(specified by parameter $\ell$), threshold $T$ and the corresponding
average reduced cost and may be unrelated to the cost function
$\bar{c}$ that we optimized to obtain the fractional solution $(x,y)$.

\begin{lemma}\label{lem:oblivious-clustering}
  Consider a feasible fractional solution $(x,y)$ to $\lpc{c}$ where
  $x$ is distance-optimal.  Let
  $\ell\geq 1$ be a positive integer, let $T\geq 0$ be arbitrary.
  Then we have
  $\expected{\cost_\ell(A)}\leq 19\ell T+19\sum_{j\in C}\cavt{T}(j)$
  where $A$ is the (random) solution output by the Algorithm~\ref{alg:main} with \emph{oblivious} clustering.
\end{lemma}
\noindent
By Lemma~\ref{lem:oblivious-clustering} we obtain that our algorithm with oblivious clustering yields a $38$-approximation.


\section{Handling the General Case}\label{sec:general-weights}

Consider an arbitrary instance of \proborderedkmedian.  Let $w$ be the
weight vector and let $\bar{w}$ the sorted weight vector using the
same weights as $w$ but without repetition. Let $R$ be the number of
distinct weight in both weight vectors.  W.l.o.g.\ we assume that all
distances $c_{ij}$ for $i\in\calF,j\in\calC$ are pairwise
distinct. (This can be achieved by slightly perturbing the input
distances.)  To apply our algorithmic framework, we guess {\it
  tresholds} $T_r$ for $r=1,\dots,R$ such that $T_r$ is the smallest
distance $c_{ij}$ that is multiplied by weight of value $\bar{w}_r$ in some
fixed optimum solution. To guess the thresholds $T_r$ we check $(nm)^R$ many
candidates.  Additionally, we define $T_0 = \infty$.  We have
$T_r < T_{r-1}$ for $r=1,\dots,R$ because we assumed pairwise distinct
distances.  For each $i\in\calF, j\in\calC$ we assign the connection cost
$c_{ij}$ to the weight $w(i,j) = \bar{w}_r$, where
$T_{r} \leq c_{ij} < T_{r-1}$.
This leads us to the following definition of our reduced cost function
$\cred_{ij} = c_{ij} \cdot w(i,j)$ for all $i\in\calF, j\in\calC$.
We compute an optimal solution $(x,y)$ to $\lpc{\cred}$ and apply 
Algorithm~\ref{alg:main} to $(x,y)$.

\begin{lemma}\label{lem:general-reduction-to-one-rect}
  The above-described algorithm is a $38$-approximation algorithm for
  \proborderedkmedian that makes $\Oh((nm)^{R})$ many calls to
  Algorithm~\ref{alg:main} with oblivious clustering, where $R$ is the
  number of distinct weights in the weight vector $w$.
\end{lemma}
\begin{proof}
  Let $A \subseteq \calF$ be the (random) solution output by the algorithm.
  Let $\OPT$ be the optimum objective function.  For each $r=1,\dots,R$ let $\lr$ be the largest index such that $w_{\lr}=\bar{w}_r$. 
  From Lemma~\ref{lem:oblivious-clustering} we have for all $r=1,\dots,R$
  \begin{equation}\label{ineq:costlrA}
   \expected{\cost_{\lr}(A)} \leq 19 \cdot \ell_r T_r + 19 \cdot \sum_{j \in \calC} \cavlr(j).
  \end{equation}
  We decompose $\cost(A)$ into rectangular ``pieces''
  (defining $w_{R+1}=0$)
  \begin{align}\label{ineq:costA}
   &\expected{\cost(A)}\hspace{-2pt}=\hspace{-2pt}\expectedBigPar{\sum_{\ell=1}^{n} w_\ell \cdot c_\ell^\to(A)}\hspace{-2pt}=\hspace{-2pt}
   \expectedBigPar{\sum_{r=1}^{R} \sum_{s=1}^{\lr} (\bar{w}_r-\bar{w}_{r+1}) \cdot c_s^\to(A)}\hspace{-2pt}=\hspace{-2pt}
   \expectedBigPar{\sum_{r=1}^{R} (\bar{w}_r-\bar{w}_{r+1}) \cdot \cost_\lr(A)} \nonumber\\
   & = \sum_{r=1}^{R} (\bar{w}_r-\bar{w}_{r+1}) \cdot \expected{\cost_\lr(A)} \stackrel{\eqref{ineq:costlrA}}{\leq}
   19 \cdot \sum_{r=1}^{R} (\bar{w}_r-\bar{w}_{r+1}) \cdot \lr T_r +
   19 \cdot \sum_{r=1}^{R} \sum_{j \in \calC} (\bar{w}_r-\bar{w}_{r+1}) \cdot \cavlr(j).
  \end{align}
  We will bound this in terms of $\OPT$.
  We know that an optimal solution pays at least cost $T_r$ for weight in $w$
  equal to $\bar{w}_r$ for $r=1,\dots,R$. Therefore, defining $\ell_0 = 0$ and $\bar{w}_{R+1} = 0$ we have
  \begin{align}\label{ineq:opt-geq-rectangles}
   &\OPT \geq \sum_{r=1}^{R} \bar{w}_r \cdot (\ell_r - \ell_{r-1}) T_r =
   \sum_{r=1}^{R} \bar{w}_r \cdot \ell_r T_r -
   \sum_{r=2}^{R} \bar{w}_r \cdot \ell_{r-1} T_r \geq\nonumber\\
   &\sum_{r=1}^{R} \bar{w}_r \cdot \ell_r T_r -
   \sum_{r=2}^{R} \bar{w}_r \cdot \ell_{r-1} T_{r-1} =
   \sum_{r=1}^{R} \bar{w}_r \cdot \ell_r T_r -
   \sum_{r=1}^{R} \bar{w}_{r+1} \cdot \ell_{r} T_{r} =
   \sum_{r=1}^{R} (\bar{w}_r - \bar{w}_{r+1}) \cdot \ell_r T_r.
  \end{align}
  Moreover, we have
  \begin{align}\label{eq:opt-eq-fractions}
   &\sum_{r=1}^{R} \sum_{j \in \calC} (\bar{w}_r-\bar{w}_{r+1}) \cdot \cavlr(j) = 
   \sum_{r=1}^{R} \sum_{j \in \calC} \sum_{i \in \calF} (\bar{w}_r-\bar{w}_{r+1}) \cdot x_{ij} \cdot c_{ij}^{T_r} = \nonumber\\
   &\sum_{j \in \calC} \sum_{i \in \calF} x_{ij} \cdot \sum_{r=1}^{R} (\bar{w}_r-\bar{w}_{r+1}) \cdot c_{ij}^{T_r} =
   \sum_{j \in \calC} \sum_{i \in \calF} x_{ij} \cdot \sum_{r\colon \bar{w}_r \leq w(i,j)}^{R} (\bar{w}_r-\bar{w}_{r+1}) \cdot c_{ij} = \nonumber\\
   &\sum_{j \in \calC} \sum_{i \in \calF} x_{ij} \cdot w(i,j) \cdot c_{ij} =
   \sum_{j \in \calC} \sum_{i \in \calF} x_{ij} \cdot \cred_{ij} \stackrel{\eqref{lp:objective}}{\leq} \OPT.
  \end{align}
  Thus, we finally have
  $\expected{\cost(A)} \stackrel{\eqref{ineq:costA},\eqref{ineq:opt-geq-rectangles},\eqref{eq:opt-eq-fractions}}{\leq} 38 \cdot \OPT$.
\end{proof}
An immediate consequence of the lemma is a constant-factor approximation algorithm for \proborderedkmedian with a constant number of different weights.

Using standard bucketing arguments and neglecting sufficiently small
weights, we can ``round'' an arbitrary weight vector into a weight
vector with only a logarithmic number of different weights losing 
a factor of $1+\epsilon$ in approximation. Plugging this into
Lemma~\ref{lem:oblivious-clustering}, we can obtain a
$(38+\epsilon)$-approximation algorithm for the general case in time
$(nm)^{\Oh(1/\epsilon\log(n))}$.

\paragraph*{Achieving Polynomial Time}

To obtain a truly-polynomial time algorithm we use the clever
bucketing approach proposed by Aouad et
al.~\cite{AouadS17:logn-apx}. In this approach the \emph{distances}
are grouped into logarithmically many distance classes thereby losing
a factor $1+\epsilon$.  For each distance class the \emph{average}
weight is guessed up to a factor of $1+\epsilon$.  The crucial point
is, that this guessing can be achieved in polynomial time because the
average weights are non-decreasing with increasing distance
class. This leads to a reduced cost function based on average
weights. The analysis of the resulting analysis decomposes the weight
vector into $n=|\calC|$ many rectangular objectives. While the proof
strategy is similar in spirit to the one of
Lemma~\ref{lem:general-reduction-to-one-rect} it turns out to be
technically more involved (see appendix for details).
\begin{theorem}
  For any $\epsilon>0$, there is an $(38+\epsilon)$-approximation
  algorithm for \proborderedkmedian with running time $(nm)^{\Oh(1/\epsilon\log(1/\epsilon))}$.
\end{theorem}

\section{Concluding Remarks and Open Questions}\label{sec:conclusions}
We have obtained a constant-factor approximation for \proborderedkmedian. We have extended the less detailed version of
the analysis of the algorithm by Charikar and
Li~\cite{CharikarL12:k-median} for \probkmedian, and hence our
constants can probably be improved. It would be interesting to see, if
our ideas can be used for other problems with ordered objectives.


%


\paragraph*{Acknowledgments.}
J. Byrka and J. Spoerhase were supported by the NCN grant number 2015/18/E/ST6/00456.
K. Sornat was supported by the NCN grant number
2015/17/N/ST6/03684.

\bibliographystyle{plain}
\bibliography{main}

\begin{thebibliography}{10}

\bibitem{AhmadianNSW16}
Sara Ahmadian, Ashkan Norouzi{-}Fard, Ola Svensson, and Justin Ward.
\newblock Better guarantees for $k$-means and {E}uclidean $k$-median by
  primal-dual algorithms.
\newblock In {\em Proc. 58th IEEE Symposium on Foundations of Computer Science
  (FOCS'17)}, 2017.

\bibitem{AlamdariS:waoa17}
Soroush Alamdari and David Shmoys.
\newblock A bicriteria approximation algorithm for the $k$-center and
  $k$-median problems.
\newblock In {\em Proceedings of the 15th International Workshop Approximation
  and Online Algorithms ({WAOA'17})}, 2017.

\bibitem{AouadS17:logn-apx}
Ali Aouad and Danny Segev.
\newblock {The ordered $k$-median problem: Surrogate models and approximation
  algorithms}.
\newblock Under review, available at
  \url{http://www.mit.edu/~aaouad/MOR-ordered-median.pdf} [2018-02-20].

\bibitem{arya2004local}
Vijay Arya, Naveen Garg, Rohit Khandekar, Adam Meyerson, Kamesh Munagala, and
  Vinayaka Pandit.
\newblock Local search heuristics for $k$-median and facility location
  problems.
\newblock {\em SIAM Journal on computing}, 33(3):544--562, 2004.

\bibitem{bartal98:tree-embeddings}
Yair Bartal.
\newblock On approximating arbitrary metrices by tree metrics.
\newblock In {\em Proc. Thirtieth Annual {ACM} Symposium on the Theory of
  Computing (STOC'98)}, pages 161--168, 1998.

\bibitem{bertsimas2014-least-quantile-regression}
Dimitris Bertsimas and Rahul Mazumder.
\newblock Least quantile regression via modern optimization.
\newblock {\em Ann. Statist.}, 42(6):2494--2525, 12 2014.

\bibitem{BertsimasS03-robust-k-sum}
Dimitris Bertsimas and Melvyn Sim.
\newblock Robust discrete optimization and network flows.
\newblock {\em Math. Program.}, 98(1-3):49--71, 2003.

\bibitem{bertsimas05:opt-over-integers}
Dimitris Bertsimas and Robert Weismantel.
\newblock {\em Optimization over integers}.
\newblock Athena Scientific, 2005.

\bibitem{BolandDNP06:exact-ordered-median}
Natashia Boland, Patricia Dom{\'{\i}}nguez{-}Mar{\'{\i}}n, Stefan Nickel, and
  Justo Puerto.
\newblock Exact procedures for solving the discrete ordered median problem.
\newblock {\em Computers {\&} {OR}}, 33(11):3270--3300, 2006.

\bibitem{BradleyFM99-math-prog-form-chall}
Paul~S. Bradley, Usama~M. Fayyad, and Olvi~L. Mangasarian.
\newblock Mathematical programming for data mining: Formulations and
  challenges.
\newblock {\em {INFORMS} Journal on Computing}, 11(3):217--238, 1999.

\bibitem{ByrkaPRST17}
Jaros{\l}aw Byrka, Thomas Pensyl, Bartosz Rybicki, Aravind Srinivasan, and Khoa
  Trinh.
\newblock An improved approximation for \emph{k}-median and positive
  correlation in budgeted optimization.
\newblock {\em {ACM} Trans. Algorithms}, 13(2):23:1--23:31, 2017.

\bibitem{ChakrabartyS17}
Deeparnab Chakrabarty and Chaitanya Swamy.
\newblock Interpolating between $k$-median and $k$-center: Approximation
  algorithms for ordered $k$-median.
\newblock {\em CoRR}, abs/1711.08715, 2017.

\bibitem{CharikarGTS99-kmedian}
Moses Charikar, Sudipto Guha, {\'{E}}va Tardos, and David~B. Shmoys.
\newblock A constant-factor approximation algorithm for the $k$-median problem.
\newblock In {\em Proc. 31st Annual {ACM} Symposium on Theory of Computing
  (STOC'99)}, pages 1--10, 1999.

\bibitem{CharikarL12:k-median}
Moses Charikar and Shi Li.
\newblock A dependent {LP}-rounding approach for the $k$-median problem.
\newblock In {\em Proc. 39th International Colloquium on Automata, Languages,
  and Programming (ICALP'12)}, pages 194--205, 2012.

\bibitem{Dominguez-MarinNHM05-heuristic}
Patricia Dom{\'{\i}}nguez{-}Mar{\'{\i}}n, Stefan Nickel, Pierre Hansen, and
  Nenad Mladenovic.
\newblock Heuristic procedures for solving the discrete ordered median problem.
\newblock {\em Annals {OR}}, 136(1):145--173, 2005.

\bibitem{DreznerN09-dc-decomposition-median}
Zvi Drezner and Stefan Nickel.
\newblock Constructing a {DC} decomposition for ordered median problems.
\newblock {\em J. Global Optimization}, 45(2):187--201, 2009.

\bibitem{DreznerN09-ordered-one-median-plane}
Zvi Drezner and Stefan Nickel.
\newblock Solving the ordered one-median problem in the plane.
\newblock {\em European Journal of Operational Research}, 195(1):46--61, 2009.

\bibitem{EspejoRV09-convex-ordered-median}
Inmaculada Espejo, Antonio~M. Rodr{\'{\i}}guez{-}Ch{\'{\i}}a, and C.~Valero.
\newblock Convex ordered median problem with $\ell_{p}$-norms.
\newblock {\em Computers {\&} {OR}}, 36(7):2250--2262, 2009.

\bibitem{FernandezPPS17-OWA-multi-obj-ST}
Elena Fern{\'{a}}ndez, Miguel~A. Pozo, Justo Puerto, and Andrea Scozzari.
\newblock Ordered weighted average optimization in multiobjective spanning tree
  problem.
\newblock {\em European Journal of Operational Research}, 260(3):886--903,
  2017.

\bibitem{harris2017symmetric}
David~G. Harris, Thomas Pensyl, Aravind Srinivasan, and Khoa Trinh.
\newblock Symmetric randomized dependent rounding.
\newblock {\em CoRR}, abs/1709.06995, 2017.

\bibitem{HochbaumS85}
Dorit~S. Hochbaum and David~B. Shmoys.
\newblock A best possible heuristic for the \emph{k}-center problem.
\newblock {\em Mathematics of Operations Research}, 10(2):180--184, 1985.

\bibitem{jain2003greedy}
Kamal Jain, Mohammad Mahdian, Evangelos Markakis, Amin Saberi, and Vijay~V
  Vazirani.
\newblock Greedy facility location algorithms analyzed using dual fitting with
  factor-revealing {LP}.
\newblock {\em Journal of the ACM (JACM)}, 50(6):795--824, 2003.

\bibitem{KalcsicsNP03-multifacility-ordered}
J{\"{o}}rg Kalcsics, Stefan Nickel, and Justo Puerto.
\newblock Multifacility ordered median problems on networks: {A} further
  analysis.
\newblock {\em Networks}, 41(1):1--12, 2003.

\bibitem{KalcsicsNPT02-algorithic-ordered}
J{\"{o}}rg Kalcsics, Stefan Nickel, Justo Puerto, and Arie Tamir.
\newblock Algorithmic results for ordered median problems.
\newblock {\em Oper. Res. Lett.}, 30(3):149--158, 2002.

\bibitem{KanungoMNPSW04}
Tapas Kanungo, David~M. Mount, Nathan~S. Netanyahu, Christine~D. Piatko, Ruth
  Silverman, and Angela~Y. Wu.
\newblock A local search approximation algorithm for $k$-means clustering.
\newblock {\em Computational Geometry}, 28(2-3):89--112, 2004.

\bibitem{labbe17-comparative-study-ordered-median}
Martine Labb{\'{e}}, Diego Ponce, and Justo Puerto.
\newblock A comparative study of formulations and solution methods for the
  discrete ordered $p$-median problem.
\newblock {\em Computers {\&} {OR}}, 78:230--242, 2017.

\bibitem{laporte-etal:location-science}
Gilbert Laporte, Stefan Nickel, and Francisco Saldanha~da Gama, editors.
\newblock {\em Location Science}.
\newblock Springer, 2015.

\bibitem{li2016approximating}
Shi Li and Ola Svensson.
\newblock Approximating $k$-median via pseudo-approximation.
\newblock {\em SIAM Journal on Computing}, 45(2):530--547, 2016.

\bibitem{Nickel:ordered-weber}
S.~Nickel.
\newblock Discrete ordered {W}eber problems.
\newblock In {\em Proc. Operations Research (OR'01)}, pages 71--76. Springer,
  2001.

\bibitem{NickelP99-location-problems-unified}
Stefan Nickel and Justo Puerto.
\newblock A unified approach to network location problems.
\newblock {\em Networks}, 34(4):283--290, 1999.

\bibitem{Nickel09-ordered-k-median-book}
Stefan Nickel and Justo Puerto.
\newblock {\em Location Theory - {A} Unified Approach}.
\newblock Springer, 2009.

\bibitem{puerto-fernandez:geometrical-props-single-facility}
J.~Puerto and F.~R. Fern{\'{a}}ndez.
\newblock Geometrical properties of the symmetrical single facility location
  problem.
\newblock {\em Journal of Nonlinear and Convex Analysis}, 1(3):321--342, 2000.

\bibitem{PuertoPG14-modified-vns}
Justo Puerto, Dionisio P{\'{e}}rez{-}Brito, and Carlos~G.
  Garc{\'{\i}}a{-}Gonz{\'{a}}lez.
\newblock A modified variable neighborhood search for the discrete ordered
  median problem.
\newblock {\em European Journal of Operational Research}, 234(1):61--76, 2014.

\bibitem{PuertoRR13-branch-cut-ordered-median}
Justo Puerto, A.~B. Ramos, and Antonio~M. Rodr{\'{\i}}guez{-}Ch{\'{\i}}a.
\newblock A specialized branch {\&} bound {\&} cut for single-allocation
  ordered median hub location problems.
\newblock {\em Discrete Applied Mathematics}, 161(16-17):2624--2646, 2013.

\bibitem{PuertoR05-exp-card-ordered-median}
Justo Puerto and Antonio~M. Rodr{\'{\i}}guez{-}Ch{\'{\i}}a.
\newblock On the exponential cardinality of {FDS} for the ordered $p$-median
  problem.
\newblock {\em Oper. Res. Lett.}, 33(6):641--651, 2005.

\bibitem{puerto2009minimax}
Justo Puerto, Antonio~M Rodr{\'\i}guez-Ch{\'\i}a, and Arie Tamir.
\newblock Minimax regret single-facility ordered median location problems on
  networks.
\newblock {\em INFORMS Journal on Computing}, 21(1):77--87, 2009.

\bibitem{PuertoRT17-k-sum}
Justo Puerto, Antonio~M. Rodr{\'{\i}}guez{-}Ch{\'{\i}}a, and Arie Tamir.
\newblock Revisiting $k$-sum optimization.
\newblock {\em Math. Program.}, 165(2):579--604, 2017.

\bibitem{PuertoT05:treeshaped-ordered}
Justo Puerto and Arie Tamir.
\newblock Locating tree-shaped facilities using the ordered median objective.
\newblock {\em Math. Program.}, 102(2):313--338, 2005.

\bibitem{Rodríguez-Chía2005:ordered-networks}
A.~M. Rodr{\'i}guez-Ch{\'i}a, J.~Puerto, D.~P{\'e}rez-Brito, and J.~A. Moreno.
\newblock The $p$-facility ordered median problem on networks.
\newblock {\em Top}, 13(1):105--126, Jun 2005.

\bibitem{StanimirovicKD07-genetic-ordered}
Zorica Stanimirovic, Jozef Kratica, and Djordje Dugosija.
\newblock Genetic algorithms for solving the discrete ordered median problem.
\newblock {\em European Journal of Operational Research}, 182(3):983--1001,
  2007.

\bibitem{Tamir01}
Arie Tamir.
\newblock {The $k$-centrum multi-facility location problem}.
\newblock {\em Discrete Applied Mathematics}, 109(3):293--307, 2001.

\bibitem{tamir98:centdiantree}
Arie Tamir, Dionisio P{\'{e}}rez{-}Brito, and Jos{\'{e}}~A.
  Moreno{-}P{\'{e}}rez.
\newblock A polynomial algorithm for the $p$-centdian problem on a tree.
\newblock {\em Networks}, 32(4):255--262, 1998.

\end{thebibliography}

\appendix

\section{Missing Proofs for Section~\ref{sec:algor-fram}}
Below, we describe some basic normalization steps for a feasible solution $(x,y)$ to \lpc{$\cred$}.
\begin{definition}
  Let $(x,y)$ be a feasible solution to \lpc{$\cred$} where $\cred$ is
  some reduced cost function.  We call the assignment $x$ of clients
  to facilities \emph{distance-optimal} if $x$ minimizes
  $\sum_{i\in\calF,j\in\calC}c_{ij}x_{ij}$ when $y$ is kept
  fixed.
\end{definition}

\begin{lemma}\label{lem:canonical-solution}
  We can w.l.o.g.\ assume that an optimal solution $(x,y)$ to
  \lpc{$\cred$} for some reduced cost function $\cred$ satisfies the
  following properties.
  \begin{enumerate}[(i)]
  \item\label{item:non-negative-facility} For any facility $i\in\calF$ we have $y_i>0$,
  \item\label{item:client-fully-assigned} for any $i\in\calF,j\in\calC$ we have $x_{ij}\in\{0,y_i\}$,
  \item\label{item:distance-optimality} the assignment $x$ is distance-optimal.
  \end{enumerate}
\end{lemma}
\begin{proof}[Proof of Lemma~\ref{lem:canonical-solution}]
  To see the third property fix the opening vector $y$ and some client
  $j$.  Now sort all facilities $i$ in non-decreasing order of their
  distance $c_{ij}$ to $j$ and greedily assign as much of the
  remaining demand of $j$ to the current facility $i$ (respecting the
  constraint $x_{ij}\leq y_i$). Stop when the full demand of $j$ is
  served and repeat this process for all clients.  Since the reduced
  cost function $\cred$ respects the order of the original distances
  (see definition) the resulting assignment is optimal also under the
  reduced cost function.

  The first and second properties are folklore and can be achieved by
  removing or duplicating facilities (see
  \cite{CharikarL12:k-median}).
\end{proof}

 \begin{proof}[Proof of Lemma~\ref{lem:lp-round-depend}]
   Charikar and Li \cite{CharikarL12:k-median} describe how to
   implement the procedure satisfying the above claims.  In doing so,
   the only requirement is that $(x,y)$ be a (not necessarily optimal)
   feasible solution to \lpc{$c$}, that the volume of each bundle
   $\calU_j$ has volume at least $1/2$ and that the union of set
   families $\{\{y_i\}\}_{i\in\cal F}$, $\{U_j\}_{j\in\calC'}$,
   $\{\calU_j\cup\calU_{j'}\mid (j,j')\in \calM\}$, $\{\calF\}$ form a
   laminar family. The latter claim on laminarity follows immediately
   from the construction in the algorithm.  The validity of
   $\vol(\calU_j)\geq 1/2$ for $j\in\calC'$ depends on the
   implementation of the clustering procedure and has thus to be
   proven for the specific implementation.
 \end{proof}

 \section{Missing Proofs for Section~\ref{sec:one_rectangle_apx}}
 \begin{proof}[Proof of Lemma~\ref{lem:filtering}]
  To see~(\ref{item:filtering-spread}) assume w.l.o.g.\ that $j$ is
  considered before $j'$ as a potential cluster center in the
  algorithm.  Thus $\cavt{T}(j)\leq\cavt{T}(j')$.  If
  $c_{jj'}\leq 4\cavt{T}(j')+4T=4\max(\cavt{T}(j),\cavt{T}(j'))+4T$ then $j'$
  would be deleted from $\calC''$ when $j$ is considered. A
  contradiction to the fact that $j'$ is a cluster center.

  In order to see~(\ref{item:filtering-assign}), consider an arbitrary
  client $j\in\calC\setminus\calC'$.  As $j$ is not a cluster center
  it was deleted from $\calC''$ when some cluster center $j'\in\calC'$
  was considered.  For this cluster center we have
  $c_{jj'}\leq 4\cavt{T}(j)+4T$.
\end{proof}

\begin{proof}[Proof of Lemma~\ref{lem:bundles}]
  To prove statement~(\ref{item:bundles-volume}) consider an arbitrary
  $j\in\cal C'$. Let $j'\in\calC'$ be such that $2R_j=c_{jj'}$.  We
  have $c_{jj'}>4\cavt{T}(j)+4T\geq 4\cav(j)$ and hence, $R_j>2\cav(j)$.
  Therefore,
  $\cav(j) = \sum_{i\in\calF_j} x_{ij}c_{ij} \geq
  \sum_{i\in\calF_j\setminus\calU_j}x_{ij}c_{ij} \geq R_j \cdot
  \sum_{i\in\calF_j\setminus\calU_j}x_{ij} \geq R_j \cdot
  \vol(\calF_j\setminus\calU_j)$ where the last inequality follows
  because $x_{ij}=y_i$ for all $i\in\calF$ and $j\in\calC'$. Therefore
  $\vol(\calF_j\setminus\calU_j) < 1/2$ and $\vol(\calU_j) > 1/2$.

  To prove~(\ref{item:bundles-disjoint}) consider distinct
  $j,j'\in\calC'$. By the definition of $R_j$ we have
  $c_{jj'}\geq 2R_j$. Hence, for any facility $i$ in $\calB(j,R_j)$ we
  have $c_{ij} < c_{ij'}$, which implies~(\ref{item:bundles-disjoint}).
\end{proof}

\section{Missing Proofs for Section~\ref{sec:oblivious-clustering}}

\begin{proof}[Proof of Lemma~\ref{lem:oblivious-clustering}]
  The idea of the proof is to provide for each client an upper bound
  on the distance $C_j$ (according to the original distance $c$) traveled by
  this client. In what follows $c_1, c_2$ are constants to be determined later.

  The upper bound is paid for by two budgets $D_j$ and $X_j$. The
  ``deterministic'' budget $D_j$ is $c_1T$. The ``probabilistic''
  budget $X_j$ is a random variable (depending on the random choices
  made by the algorithm).

  We will show below that (by a suitable choice of $X_j$) the
  connection cost $C_j$ of $j$ can actually be upper bounded by
  $D_j+X_j$ and $\expected{X_j} \leq c_2\cavt{T}(j)$. If this can be
  shown then this will complete our proof of a constant-factor
  approximation.  To this end note that at most $\ell$ clients $j$
  will pay the budget $D_j=c_1T$ because at most $\ell$ distances are
  actually accounted for in the objective function. Analogously to the case of dedicated clustering, we obtain:
\begin{displaymath}
  \expected{\ALG} \leq D_j \cdot \ell + \sum_{j \in C} \expected{X_j} \leq
  c_1\ell\cdot T+c_2 \cdot \sum_{j\in C}\cavt{T}(j).
\end{displaymath}

To show the claim consider an arbitrary client $j$ with connection
cost $C_j$. We incrementally construct our upper bound on $C_j$
starting with $0$.  Each increment will be either charged to $D_j$ or
$X_j$.

Consider a client $j$ and the cluster center $j'$ it is assigned to
(possibly $j=j'$).  We have that $c_{jj'}\leq 4\cav(j)\leq 4\cavt{T}(j)+4T$ by
line~\ref{alg:line:clustering} of Algorithm~\ref{alg:oblivious-clustering}. We charge
$4T$ to $D_j$ and $4\cavt{T}(j)$ with probability~1 to $X_j$.

We now describe how to pay for the transport from $j'$ to an open
facility. There are two cases to distinguish. Either a facility
within a radius of $\beta T$ is opened or not. (Here, $\beta\geq 2$ is a parameter to be determined later.) If yes, then this cost can be
covered by charging an additional amount of $\beta T$ to $D_j$.
In this case the total cost is upper bounded by $D_j = (\beta+4) T$ and
$\expected{X_j} = 4\cavt{T}(j)$.

If \emph{no} facility within a radius of $\beta T$ of $j'$ is opened then observe that for
all facilities $i$ with $c_{ij'}\geq \beta T$ we have that
$c^T_{ij'}=c_{ij'}$ because of $\beta \geq 1$.  We now continue to bound the connection cost for
this case.  Let $j''$ be the closest client distinct from $j'$ in
$\calC'$. We consider the case where $j''$ and $j'$ are not
matched. (The case where they are matched is simpler.) Let $j'''$ be
the client in $\calC'$ to which $j''$ is matched i.e. $(j'',j''') \in \calM$.  By the dependent
rounding process one facility in $\calU_{j''}\cup\calU_{j'''}$ will be
opened. We have that $c_{j'j''}=2R_{j'}=2R$ and thus $c_{j''j'''}\leq 2R$ and $R_{j''},R_{j'''}\leq R$
(otherwise, $j''$ and $j'''$ would not have been matched).
This means that in case no facility is opened in the bundle $\calU_{j'}$
the client $j$ travels an additional distance (in expectation) of at most
$\max(c_{j'j''} + R_{j''}, c_{j'j''} + c_{j''j'''} + R_{j'''}) \leq 2R + 2R + R\leq 5R$.

If a facility is opened in the bundle $\calU_{j'}$ then we charge this
additional connection cost to $X_j$. The contribution of the additional connection cost in this case to the expectation of $X_j$ cost is at most
\begin{equation}\label{eq:cfar}
  \sum_{i \in \calU_{j'} \setminus \calB(j',\beta T)} x_{ij'}c_{ij'} =
  \sum_{i \in \calU_{j'} \setminus \calB(j',\beta T)} x_{ij'}c^T_{ij'} \leq
  \sum_{i \in \calF_{j'} \setminus \calB(j',\beta T)} x_{ij'}c^T_{ij'}\,.
\end{equation}
Here, equality holds because we assume that no facility is opened in $\calB(j',\beta T)$ where $\beta \geq 1$ and because therefore $c_{ij'}=c^T_{ij'}$ for all $i\in\calU_{j'}\setminus B(j',\beta T)$.  The right hand side of~(\ref{eq:cfar}) is denoted by  $\cfarr(j')$  and is clearly upper bounded by $\sum_{i \in \calF_{j'}} x_{ij}c^T_{ij}=\cavt{T}(j')$.

We finally handle the case where no facility in $\calU_{j'}$ is opened
and where $j$ additionally travels a distance of at most $5R$. If $R\leq \beta T$, we can charge the additional travel distance of at most $5\beta T$ to $D_j$. Hence, we focus on the difficult case where $R>\beta T$ and where the maximum distance traveled can be unbounded in terms of $T$.  We
charge this additional cost to $X_j$. We bound the probability that
this case occurs.  We claim that $\vol(\calU_{j'})$ is at least
$1-\cfarr(j')/R$.  To see this, note that for all facilities in
$\calF_{j'}\setminus\calU_{j'}$ we have that $c^T_{ij'}=c_{ij}$ because $R>\beta T$ and $\beta \geq 1$.
Hence
\begin{align*}
 \cfarr(j') & = \sum_{i \in \calF_{j'} \setminus \calB(j',\beta T)} \hspace{-15pt} x_{ij'} \cdot c^T_{ij'} =
 \sum_{i \in \calF_{j'} \setminus \calB(j',\beta T)} \hspace{-15pt} x_{ij'} \cdot c_{ij'} \geq
 \sum_{i \in \calF_{j'} \setminus \calU_{j'}} \hspace{-10pt} x_{ij'} \cdot c_{ij'} \\
 &\geq R \cdot \hspace{-10pt} \sum_{i \in \calF_{j'} \setminus \calU_{j'}} x_{ij'} =
 R \cdot \vol(\calF_{j'} \setminus \calU_{j'}) =
 R \cdot (1-\vol(\calU_{j'})),
\end{align*}
which implies the claim. Here, note that $\calB(j',\beta T)  \subseteq \calU_{j'}$ because $R>\beta T$.
This means that $j$ travels the additional distance of at most $5 R$ with probability
at most $\cfarr(j')/R$ and hence the increment to $X_j$ in expectation is
upper bounded by $5 \cdot \cfarr(j')$.
Thus, for the case where no facility is opened within a radius of $\beta T$ around $j'$, we can upper bound
$\expected{X_j}$ by:
\begin{itemize}
 \item an expected cost of serving client $j$ through
 the closest cluster center $j'$ that is $4 \cdot \cavt{T}(j)$ (random part), plus
 \item a value $\cfarr(j')$ (with probability at most one), plus
 \item a value $5 \cdot R$ with probability at most $\cfarr(j')/R$.
\end{itemize}
Hence $\expected{X_j} \leq 4 \cdot \cavt{T}(j) \cdot 1 + \cfarr(j') \cdot
1 + 5 \cdot R \cdot \cfarr(j')/R=4\cavt{T}(j)+5\cfarr(j')$.
As in the oblivious clustering we sort the clients according to $\cav$ rather than $\cavt{T}$ we do not necessarily have that $\cavt{T}(j')$ or even $\cfarr(j')$ are upper bounded by $\cavt{T}(j)$.  We still can relate the latter two quantities in the following way.

First, assume that $c_{jj'}>\alpha T$ where $1\leq \alpha<\beta-1$ is a parameter to be determined later. We have that $\cav(j')\leq \cav(j)$ by our (oblivious) clustering.  Hence $\cavt{T}(j')\leq\cav(j')\leq\cav(j)\leq\cavt{T}(j)+T$. On the other hand, $\alpha T< c_{jj'}\leq 4\cav(j)$ since $j$ was assigned to $j'$.  Hence $T<4/\alpha\cdot \cav(j)$ and thus $\cavt{T}(j')\leq (1+4/\alpha)\cavt{T}(j)$. Since $\cfarr(j')\leq \cavt{T}(j')$ we can upper bound $\expected{X_j}$ in this case by $(9+20/\alpha)\cavt{T}(j)$.

Second, assume that $c_{jj'}\leq \alpha T$. Recall that we assume further that no facility is opened within $\calB(j',\beta T)$. We claim that in the assignment vector $x$ the total demand assigned from $j'$ to $\calF\setminus \calB(j',\beta T)$ is \emph{at most} the total demand assigned from $j$ to $\calF\setminus \calB(j',\beta T)$.  This is, because any facility within the ball $\calB(j',\beta T)$ is (trivially) strictly closer than any facility not in this ball.  Hence, if $j$ would manage to assign strictly more demand to facilities inside the ball than $j'$ does, then we could construct a new assignment for $j'$ that also serves strictly more demand of $j'$ within this ball contradicting the optimality of $x$.  Now, we are going to construct a (potentially suboptimal) assignment of the part of the demand of $j'$ contributing to $\cfarr(j')$ that can be upper bounded in terms of $\cavt{T}(j)$. As the optimum assignment will clearly will have the same upper bound this will conclude our proof.  To this end, we simply assign the demand of $j'$ outside of the ball $\calB(j',\beta T)$ in the same way as does $j$.  Note that by our above claim this provides enough demand as $j$ ships at least as much demand outside the ball as $j'$ does.  In particular let $i$ be an arbitrary facility outside the ball.  We now set $x_{ij'}':=x_{ij}$ to obtain our new assignment for $j'$
Note that by triangle inequality $c_{ij}\geq c_{ij'}-c_{jj'}\geq (\beta-\alpha)T\geq T$ and thus $c_{ij}=c^T_{ij}$
(a constraint $\alpha \leq \beta - 1$ was introduced to obtain $c_{ij} = c^T_{ij}$ in this case).
Therefore
\begin{displaymath}
  \frac{c^T_{ij'}}{c^T_{ij}}=\frac{c_{ij'}}{c_{ij}}\leq \frac{c_{ij'}}{c_{ij'}-c_{jj'}}\leq\frac{\beta T}{(\beta-\alpha)T}=\frac{\beta}{\beta-\alpha}\,.
\end{displaymath}
$x'$ can be not optimal assignement for $j'$, hence
\begin{displaymath}
  \cfarr(j') \leq \sum_{i \in \calF \setminus \calB(j',\beta T)} x_{ij'}'c^T_{ij'} \leq \frac{\beta}{\beta-\alpha}\sum_{i \in \calF \setminus \calB(j',\beta T)} x_{ij}c^T_{ij}\leq \frac{\beta}{\beta-\alpha}\cavt{T}(j)\,.
\end{displaymath}
In the end we have two upper bounds for both budgets:
\begin{align*}
  D_j &\leq (4+5\beta) T,\\
  \expected{X_j} &\leq  \max\left\{4+\frac{5\beta}{\beta-\alpha},9+\frac{20}{\alpha}\right\} \cavt{T}(j).
\end{align*}
Plugging $\alpha=2$ and $\beta=3$ gives the desired constants in the claim.
\end{proof}

\section{Missing Proofs from Section~\ref{sec:general-weights}}\label{sec:proofs-general-weights}

\begin{corollary}
 Let $\calI$ will be an instance of \proborderedkmedian
 with constant number of different weights in $w$.
 There exists a randomized algorithm that solves \proborderedkmedian on $\calI$
 and gives $38$-approximation (in expectation) in polynomial time.
\end{corollary}
\begin{proof}
 We have $\Oh(|\{ w_j: \: j \in \{1, 2, \dots, n\} \}|) = \Oh(1)$. 
 Therefore using Lemmas~\ref{lem:general-reduction-to-one-rect} and
 \ref{lem:oblivious-clustering} we get $38$-approximation solution in
 $n^{\Oh(1)} \cdot t(\calA) = n^{\Oh(1)} \cdot \poly(nm) = \poly(nm)$ time.
\end{proof}

 In Lemma~\ref{lem:gen-weights-to-log-eps-n-stairs} we show
 how to reduce the number of different weights to
 at most $\Oh(\log_{1+\epsilon}{n})$.
 Main idea of such reduction is partitioning an interval $[0,w_1]$ into buckets
 with geometrical step $1+\epsilon$. Solving such instance we lose
 factor $1+\epsilon$ on approximation because
 for $\alpha$-approximation solution $W_\alpha^*$
 for $\calI^*$, optimal solution $W_{\OPT}^*$ for $\calI^*$ and optimal solution $W_{\OPT}$ for $\calI$ we have
 \begin{align*}
  \cost_{\calI}(W_\alpha^*) &\leq (1+\epsilon) \cost_{\calI^*}(W_\alpha^*) \leq (1+\epsilon) \alpha \cdot \cost_{\calI^*}(W_{\OPT}^*)
 \leq (1+\epsilon) \alpha \cdot \cost_{\calI^*}(W_{\OPT}) \leq \\
 & \leq (1+\epsilon) \alpha \cdot \cost_{\calI}(W_{\OPT}). 
 \end{align*}

\begin{theorem}\label{thm:general-weights-const-apx}
 Let $\calA$ be an algorithm for \probonerectangleorderedkowa described
 in Section~\ref{sec:oblivious-clustering} (Algorithm~\ref{alg:main} with
 oblivious clustering, i.e., Algorithm~\ref{alg:oblivious-clustering}).
 Let $t(\calA)$ be running time of $\calA$.
 Then for any $\epsilon > 0$ there exists
 an $(38+\epsilon)$-approximation algorithm for \proborderedkmedian
 that runs in $(nm)^{\Oh(\log_{1+\epsilon}{n})} \cdot t(\calA)$ time.
\end{theorem}
\begin{proof}

 We change a vector of weights $w$ into $w^*$ using Lemma~\ref{lem:gen-weights-to-log-eps-n-stairs}
 and use Lemma~\ref{lem:general-reduction-to-one-rect}.

\end{proof}

\begin{corollary}
 There exists a randomized algorithm for \proborderedkmedian that gives
 $(38+\epsilon)$-approximation (in expectation) in
 quasi-polynomial time $n^{\Oh\left( 1/\epsilon \log{n} \right)} \cdot \poly(m)$
 for any $\epsilon > 0$.
\end{corollary}
\begin{proof}
 Using Theorem~\ref{thm:general-weights-const-apx} and
 Lemma~\ref{lem:oblivious-clustering}
 we get an $38(1+\epsilon)$-approximation solution in
 $n^{\Oh(\log_{1+\epsilon}{n})} \cdot \poly(m) =
 n^{\Oh\left( 1/\epsilon \log{n} \right)} \cdot \poly(m)$ time.
\end{proof}

\begin{lemma}\label{lem:gen-weights-to-log-eps-n-stairs}
 Let $\calI = (\calF, \calC, c, k, w)$ be an instance of \proborderedkmedian and $\epsilon > 0$.
 There exists an instance $\calI^* = (\calF, \calC, c, k, w^*)$ of \proborderedkmedian
 such that for any solution $\calW \subseteq \calF, |\calW| = k$ we have
 \[  \cost_{\calI^*}(\calW) \leq \cost_{\calI}(\calW) \leq (1+\epsilon) \cdot \cost_{\calI^*}(\calW)\]
 and $w^*$ has at most $\Oh(\log_{1+\epsilon} n)$ different values, i.e.,
 $|\{a: \exists j \quad w_j^* = a \}| \in \Oh(\log_{1+\epsilon} n)$.
\end{lemma}
\begin{proof}
 We define $w^*$ by
\begin{equation}
 w_j^* = 
 \begin{cases}
 w_1 \quad &\text{for} \quad j = 1,\\
 (1+\epsilon)^{\left\lfloor \log_{1+\epsilon}{w_j} \right\rfloor} \quad &\text{for} \quad w_j > \frac{\epsilon w_1}{n} \;\text{and}\; j \neq 1 ,\\
 0 \quad &\text{for} \quad w_j \leq \frac{\epsilon w_1}{n}.\\
 \end{cases}
\end{equation}
First inequality follows directly from the definition of $w_j^*$. For the second inequality we have
\begin{align*}
 \cost_{\calI}(\calW) =
 w_1 \cdot c_1^\to(\calW) + &\sum_{\substack{j: \: w_j > \frac{\epsilon w_1}{n}\\j \neq 1}} w_j \cdot c_j^\to(\calW) + \sum_{\substack{j: \: w_j \leq \frac{\epsilon w_1}{n}\\j \neq 1}} w_j \cdot c_j^\to(\calW) \leq \\
 w_1^* \cdot c_1^\to(\calW) + &\sum_{\substack{j: \: w_j > \frac{\epsilon w_1}{n}\\j \neq 1}} (1+\epsilon) \cdot w_j^* \cdot c_j^\to(\calW) + \sum_{\substack{j: \: w_j \leq \frac{\epsilon w_1}{n}\\j \neq 1}} \frac{\epsilon w_1}{n} \cdot c_j^\to(\calW) \leq \\
 w_1^* \cdot c_1^\to(\calW) + &\sum_{\substack{j: \: w_j > \frac{\epsilon w_1}{n}\\j \neq 1}} (1+\epsilon) \cdot w_j^* \cdot c_j^\to(\calW) + \epsilon \cdot w_1^* \cdot c_1^\to(\calW) = \\
 (1+\epsilon) \cdot &\sum_{j: \: w_j > \frac{w_1}{n}} w_j^* \cdot c_j^\to(\calW) = (1+\epsilon) \cdot \cost_{\calI^*}(\calW).
\end{align*}
Let us assume that there is at least $2 \log_{1+\epsilon}{n} + 5$ different values $w_j^*$ and $n$ is large enough.
We know that the highest value of $w_j^*$ is equal to $w_1$.
It is possible that the lowest value of $w_j^*$ is equal to $0$.
By the induction we can show that the $p$-th highest value of
$\{ w_j^*: w_j^* > 0 \}$ for $p \in \{2, 3, \dots, \lceil2\log_{1+\epsilon}{n}\rceil+3\}$ is at most $\frac{w_1}{(1+\epsilon)^{p-2}}$.
Therefore 
\[ \min\{ w_j^*: w_j^* > 0 \} < \frac{w_1}{(1+\epsilon)^{\lceil2\log_{1+\epsilon}{n}\rceil+3-2}} \leq \frac{w_1}{(1+\epsilon) n^2} \leq \frac{w_1}{(1+\epsilon) n}.\]
So there exists $j$ such that $\frac{w_1}{n} \geq (1+\epsilon) w_j^* > 0$ and $w_j > \frac{\epsilon w_1}{n}$.
From the definition we have
$(1+\epsilon) w_j^* = (1+\epsilon)^{\left\lfloor \log_{1+\epsilon}{w_j} \right\rfloor + 1} > (1+\epsilon)^{ \log_{1+\epsilon}{w_j}} = w_j > \frac{\epsilon w_1}{n} > \frac{w_1}{n}$. Contradiction.
Therefore it is not possible that $w^*$ has at least $2 \log_{1+\epsilon}{n} + 5$ different values.
It means that $w^*$ has at most $\Oh(\log_{1+\epsilon}{n})$ different values.
 
\end{proof}

\section{Details for Polynomial-Time Algorithm from Section~\ref{sec:general-weights}}
We apply the clever bucketing approach of Aouad and Segev~\cite{AouadS17:logn-apx}, which allows guessing an
approximate reduced cost function in polynomial time.  Applying this
idea within our LP-based framework, however, renders the analysis more
intricate than the one with quasi-polynomial time.

\subsection{Distance Bucketing}
Let $\calW_\OPT$ be an optimal solution to a given \proborderedkmedian
instance. Let $c_{\max}:=c^\to_1(\calW_\OPT)$ be the maximum connection
cost in this solution. We assume that we know $c_{\max}$ as it is one
of $\Oh(mn)$ many possible distances in the input.  Fix an error
parameter $\epsilon>0$ and let $c_{\min}:=\epsilon\cdot
c_{\max}/n$. Roughly speaking, distances smaller than $c_{\min}$ can
have only negligible impact on any feasible solution as they may increase
its cost by a factor of at most $1+\epsilon$.

We now partition the distances of the vector $c^\to(\calW_\OPT)$ into
$S:=\lceil\log_{1+\epsilon}(n/\epsilon)\rceil=\Oh(\frac1\epsilon\log\frac{n}{\epsilon})$
many distance classes.  More precisely, for all $s=0,\dots,S-1$ introduce
the intervals
$D_s=(c_{\max}(1+\epsilon)^{-(s+1)},c_{\max}(1+\epsilon)^{-s}]$. Let
$D_S=[0,c_{\max}(1+\epsilon)^{-S}]\ni c_{\min}$.  For all
$s=0,\dots,S$ let $J_s=\{\,j\mid c^{\to}_j(\calW_\OPT)\in D_s\,\}$ and
let $C_s=\{\,c^{\to}_j(\calW_\OPT)\mid j\in J_S\,\}$. The classes
$C_0,\dots,C_S$ form a disjoint partition of $c^\to(\calW_\OPT)$ where
some of the classes may, however, be empty.  For technical reasons, we
assume that none of the intput distances $c_{ij},i\in\calF,j\in\calC$
coincides with a boundary of one of the intervals
$D_s$ for some $s=0,\dots,S$.  This can be achieved by slightly increasing all
boundaries of the intervals using the fact that the intervals are
left-open.
Additionally we define $J_{\geq s} = \bigcup_{r=s}^S J_r$.

\subsection{Guessing Average Weights}
For any non-empty class $C_s$ let
\begin{equation}\label{def:wavs}
 \wav^s:=\frac{1}{|C_s|}\Sum_{j\in J_s}w_j
\end{equation}
denote the \emph{average} weight applied to distances in this class.
If $C_s$ is empty then $\wav^s$ denotes the smallest weight $w_j$
applied to some distance $c^{\to}_j(\calW_\OPT)$ in a non-empty
class $C_l$ with $l<s$.  Such a class always exists as
$C_0\ni c_{\max}$ is non-empty.

As argued by Aouad and Segev~\cite{AouadS17:logn-apx}, it is possible
to guess the values of $\wav^s$ up to a factor of $1+\epsilon$ in
polynomial time $n^{\Oh(1/\epsilon\log{1/\epsilon})}$.  This is, because
we have $\wav^0\geq\wav^1\geq\dots\geq\wav^S$ and because it suffices
to guess those values as powers of $1+\epsilon$. More precisely, as a
result of this we assume that we are given values $\wg^0\geq\wg^1\geq\dots\geq\wg^S$ with
$\wav^s\leq \wg^s\leq (1+\epsilon)\wav^s$ for $i=0,\dots,S$.

\subsection{Reduced Cost Function and LP-Solving}
We are now ready to define our reduced cost function. 
For all values of $d\in [0,c_{\max}]$ let $w(d)$ be the weight $\wg^s$ such that
$d\in D_s$ for some $s\in\{0,\dots,S\}$.
For each $i\in\calF, j\in\calC$ such that
$c_{ij}\leq c_{\max}$ let $\cred_{ij}:= w(c_{ij})\cdot c_{ij}$.
Now solve the linear program \lpc{$\cred$} with additional constraints $x_{ij}=0$ for all
$i\in\calF,j\in\calC$ such that $c_{ij}>c_{\max}$.
In what follows let $(x,y)$ denote an optimal solution to this LP.
Now apply the rounding algorithm of Charikar and Li with oblivious clustering
(Algorithm~\ref{alg:main} with clustering as in
Algorithm~\ref{alg:oblivious-clustering}) to obtain an integral solution
$A \subseteq \calF, |A|=k$.

  Let $\OPT$ be the value (cost) of an optimum solution $W_\OPT$ for \proborderedkmedian and let
  $\OPTLP$ be the value of an optimum solution for $\lpc{\cred}$,
  let $A$  be the solution for \proborderedkmedian computed by our algorithm.
  We define distance class (interval) in which the distance $d$ falls by $D(d)$
  and $w_{n+1}=0$.

  Using Lemma~\ref{lem:oblivious-clustering} with $T_\ell = \max(D(c_\ell^\to(W_\OPT)))$
  for each $\ell=1,\dots,n$ we obtain
  \begin{equation}\label{ineq:costlA}
   \expected{\cost_{\ell}(A)} \leq c_1 \cdot \ell T_\ell + c_2 \cdot \sum_{j \in \calC} \cav^{T_\ell}(j)\,.
  \end{equation}
  We can partition the cost $\cost(A)$ of our algorithm as follows into rectangular pieces
  \begin{align}\label{ineq:costA-gen}
   &\expected{\cost(A)} = \expectedBigPar{\sum_{\ell=1}^{n} w_\ell \cdot c_\ell^\to(A)} = 
   \expectedBigPar{\sum_{\ell=1}^{n} \sum_{r=1}^{\ell} (w_\ell-w_{\ell+1}) \cdot c_r^\to(A)} = \nonumber\\
   &\expectedBigPar{\sum_{\ell=1}^{n} (w_\ell-w_{\ell+1}) \cdot \cost_\ell(A)} \stackrel{\eqref{ineq:costlA}}{\leq}
   19 \cdot \sum_{\ell=1}^{n} (w_\ell-w_{\ell+1}) \cdot \ell T_\ell +
    19 \cdot \sum_{\ell=1}^{n} \sum_{j \in \calC} (w_\ell-w_{\ell+1}) \cdot \cav^{T_\ell}(j)\,.
  \end{align}
  We would like to upper bound this in terms of $\OPT$.
  We know that the optimal solution pays at least cost $\inf(D_s)$ for each
  distance in distance bucket $C_s$ and thus
  \begin{align}
   \OPT \geq &\sum_{s=0}^{S} \left( \inf(D_s) \cdot \sum_{\ell \in J_s} w_\ell \right) =
   \sum_{s=0}^{S} \inf(D_s) \cdot \wav^s \cdot |C_s|. \label{ineq:opt-geq-rectangles-av}
  \end{align}
  \begin{lemma}\label{lem:ineq-wavs-geq-wls}
  \begin{equation}\label{ineq:wavs-geq-wlltl}
   \sum_{s=0}^{S} \inf(D_s) \cdot \wav^s \cdot |C_s| \geq
   \frac{1}{1+\epsilon} \sum_{\ell=1}^{n} (w_\ell-w_{\ell+1}) \cdot \ell \cdot T_\ell.
  \end{equation}
  \end{lemma}
  \begin{proof}
  The right hand side is equal to
  \begin{align*}
   &\frac{1}{1+\epsilon} \sum_{\ell=1}^{n} (w_\ell-w_{\ell+1}) \cdot \ell \cdot \max(D(c_\ell^\to(W_\OPT))) \leq \\
   &\sum_{\ell=1}^{n} (w_\ell-w_{\ell+1}) \cdot \ell \cdot \inf(D(c_\ell^\to(W_\OPT))) =
   \sum_{s=0}^S \sum_{\ell \in J_s} (w_\ell-w_{\ell+1}) \cdot \ell \cdot \inf(D_s) = \\
   &\sum_{\substack{s=0\\J_s \neq \emptyset}}^S \inf(D_s)\left( w_{\min(J_s)} \cdot \min(J_s)
   + \hspace{-10pt}\sum_{\ell \in J_s \setminus \min(J_s)} \hspace{-10pt} w_\ell \cdot \ell
   - \hspace{-10pt} \sum_{\ell \in J_s \setminus \max(J_s)} \hspace{-10pt} w_{\ell+1} \cdot \ell
   - w_{\max(J_s)+1} \cdot \max(J_s) \right) = \\
   &\sum_{\substack{s=0\\J_s \neq \emptyset}}^S \inf(D_s)\left( w_{\min(J_s)} \cdot \min(J_s)
   + \hspace{-10pt} \sum_{\ell \in J_s \setminus \min(J_s)} \hspace{-10pt} w_\ell \cdot \ell
   - \hspace{-13pt} \sum_{\ell \in J_s \setminus \min(J_s)} \hspace{-15pt} w_{\ell} \cdot (\ell-1)
   - w_{\max(J_s)+1} \cdot \max(J_s) \right) = \\
   &\sum_{\substack{s=0\\J_s \neq \emptyset}}^S \inf(D_s)\left( w_{\min(J_s)} \cdot (\min(J_s)-1)
   + \sum_{\ell \in J_s} w_\ell
   - w_{\max(J_s)+1} \cdot \max(J_s) \right) = \\
   &\sum_{s=0}^S \inf(D_s) \cdot \wav^s \cdot |C_s| +
   \sum_{\substack{s=0\\J_s \neq \emptyset}}^S \inf(D_s) \left( w_{\min(J_s)} \cdot (\min(J_s)-1)
   - w_{\max(J_s)+1} \cdot \max(J_s) \right).
  \end{align*}
  Proof ends with showing that the second factor is non-positive:
  \begin{align}
   &\sum_{\substack{s=0\\J_s \neq \emptyset}}^S \inf(D_s) \left( w_{\min(J_s)} \cdot (\min(J_s)-1)
   - w_{\max(J_s)+1} \cdot \max(J_s) \right) = \nonumber\\
   &\sum_{\substack{s=1\\J_s \neq \emptyset}}^S \inf(D_s) \cdot w_{\min(J_s)} \cdot (\min(J_s)-1) -
   \sum_{\substack{s=0\\J_s \neq \emptyset}}^{S-1} \inf(D_s) \cdot w_{\max(J_s)+1} \cdot \max(J_s) = \nonumber\\
   &\sum_{\substack{s=1\\J_s \neq \emptyset}}^S \inf(D_s) \cdot w_{\min(J_s)} \cdot (\min(J_s)-1) -
   \sum_{\substack{s=1\\J_{s-1} \neq \emptyset}}^S \inf(D_{s-1}) \cdot w_{\max(J_{s-1})+1} \cdot \max(J_{s-1}) = \nonumber\\
   &\sum_{\substack{s=1\\J_s \neq \emptyset}}^S \inf(D_s) \cdot w_{\min(J_{\geq s})} \cdot (\min(J_{\geq s})-1) -
   \sum_{\substack{s=1\\J_{s-1} \neq \emptyset}}^S \inf(D_{s-1}) \cdot w_{\min(J_{\geq s})} \cdot (\min(J_{\geq s})-1) = \label{eq:sum-merging}\\
   &\sum_{\substack{s=1\\J_{s-1} \neq \emptyset\\J_s \neq \emptyset}}^S (\inf(D_s) - 
   \inf(D_{s-1})) \cdot w_{\min(J_{\geq s})} \cdot (\min(J_{\geq s})-1) + \nonumber\\
   &\sum_{
   \substack{0 \leq s_1 < s_2 \leq S\\
   s_1+1 < s_2\\
   J_{s_1}, J_{s_2} \neq \emptyset\\
   J_{s_3} = \emptyset \text{ for } s_1 < s_3 < s_2 }}\hspace{-13pt}
   \inf(D_{s_2}) \cdot w_{\min(J_{\geq s_2})} \cdot (\min(J_{\geq s_2})-1)-
   \inf(D_{s_1}) \cdot w_{\min(J_{\geq s_1+1})} \cdot (\min(J_{\geq s_1+1})-1) \leq \nonumber\\
   &\sum_{
   \substack{0 \leq s_1 < s_2 \leq S\\
   s_1+1 < s_2\\
   J_{s_1}, J_{s_2} \neq \emptyset\\
   J_{s_3} = \emptyset \text{ for } s_1 < s_3 < s_2 }}\hspace{-13pt}
   \inf(D_{s_2}) \cdot w_{\min(J_{\geq s_2})} \cdot (\min(J_{\geq s_2})-1)-
   \inf(D_{s_1}) \cdot w_{\min(J_{\geq s_2})} \cdot (\min(J_{\geq s_2})-1) = \nonumber\\
   &\sum_{
   \substack{0 \leq s_1 < s_2 \leq S\\
   s_1+1 < s_2\\
   J_{s_1}, J_{s_2} \neq \emptyset\\
   J_{s_3} = \emptyset \text{ for } s_1 < s_3 < s_2 }}\hspace{-13pt}
   (\inf(D_{s_2}) - \inf(D_{s_1}))
   \cdot w_{\min(J_{\geq s_2})} \cdot (\min(J_{\geq s_2})-1) \leq 0.\nonumber
  \end{align}
  The equality~\eqref{eq:sum-merging} is just a merge of two sums into two cases:
  when two consecutive class $C_{s-1}, C_s$ are non-empty or there are some positive number
  of empty classes between two non-empty classes $C_{s_1}, C_{s_2}$.
  \end{proof}
  \noindent
  For the second term from~\eqref{ineq:costA-gen} we have
  \begin{align}\label{eq:opt-ineq-fractions-lose-eps}
   &\sum_{\ell=1}^{n} \sum_{j \in \calC} (w_\ell-w_{\ell+1}) \cdot \cav^{T_\ell}(j) = \nonumber\\
   &\sum_{\ell=1}^{n} \sum_{j \in \calC} \sum_{i \in \calF} (w_\ell-w_{\ell+1}) \cdot x_{ij} \cdot c_{ij}^{T_\ell} = \nonumber\\
   &\sum_{j \in \calC} \sum_{i \in \calF} x_{ij} \cdot \sum_{\ell=1}^{n} (w_\ell-w_{\ell+1}) \cdot c_{ij}^{T_\ell} \stackrel{\eqref{def-cijT}}{=} \nonumber\\
   &\sum_{j \in \calC} \sum_{i \in \calF} x_{ij} \cdot \sum_{\substack{\ell = 1\\ \ell\colon c_{ij} > T_\ell}}^n (w_\ell-w_{\ell+1}) \cdot c_{ij} = \nonumber\\
   &\sum_{s=0}^S \sum_{\substack{j \in \calC, i \in \calF\\ c_{ij} \in C_s}} x_{ij} \cdot c_{ij} \cdot
    \sum_{\substack{\ell = 1\\\ell \in J_{\geq s+1}}}^n (w_\ell-w_{\ell+1}) = \nonumber\\
   &\sum_{s=0}^S \sum_{\substack{j \in \calC, i \in \calF\\ c_{ij} \in C_s}} x_{ij} \cdot c_{ij} \cdot
   w_{\min\{ J_{\geq s+1}\} } \leq \nonumber\\ 
   &\sum_{s=0}^S \sum_{\substack{j \in \calC, i \in \calF\\ c_{ij} \in C_s}} x_{ij} \cdot c_{ij} \cdot \wav^s \leq
   (1+\epsilon) \sum_{s=0}^S \sum_{\substack{j \in \calC, i \in \calF\\ c_{ij} \in C_s}} x_{ij} \cdot c_{ij} \cdot \wg^s = \nonumber\\
   &(1+\epsilon) \sum_{j \in \calC} \sum_{i \in \calF} x_{ij} \cdot \cred_{ij} \stackrel{\eqref{lp:objective}}{=} (1+\epsilon) \cdot \OPTLP.
  \end{align}
  This we can upper bound in terms of value of optimal solution $\OPT$.
  For that let us define the optimal solution $W_\OPT$
  as a feasible solution of $\lpc{\cred}$ and denote it as $(x^\OPT,y^\OPT)$.
  It means that $y_i^\OPT = 1 \iff i \in W_\OPT$ and $y_i^\OPT = 0 \iff i \notin W_\OPT$.
  \begin{align}\label{ineq:optlp-leq-opt}
   \OPTLP \leq &\sum_{j \in \calC} \sum_{i \in \calF} x_{ij}^\OPT \cred_{ij} =
   \sum_{j \in \calC} \sum_{i \in \calF} x_{ij}^\OPT c_{ij} w(c_{ij}) =
   \sum_{s=0}^S \sum_{\substack{j \in \calC, i \in \calF\\ c_{ij} \in C_s}} x_{ij}^\OPT c_{ij} \wg^s \leq \nonumber\\
   &\sum_{s=0}^S \left( \max(D_s) \cdot \wg^s \sum_{\substack{j \in \calC, i \in \calF\\ c_{ij} \in C_s}} x_{ij}^\OPT \right) = \nonumber\\
   &\sum_{s=0}^S \max(D_s) \cdot \wg^s \cdot |C_s| \leq
   (1+\epsilon)^2 \cdot \sum_{s=0}^S \inf(D_s) \cdot \wav^s \cdot |C_s| = \nonumber\\
   &(1+\epsilon)^2 \cdot \sum_{s=0}^S \sum_{\ell\in J_s} w_\ell \cdot \inf(D_s) \leq (1+\epsilon)^2 \cdot \OPT.
  \end{align}
   In the end we have
   \begin{equation}
    \cost(A) \stackrel{\eqref{ineq:costA-gen},\eqref{ineq:opt-geq-rectangles-av},\eqref{ineq:wavs-geq-wlltl},
    \eqref{eq:opt-ineq-fractions-lose-eps},\eqref{ineq:optlp-leq-opt}}{\leq}
    (1+\epsilon)^3\cdot 38 \cdot \OPT.
   \end{equation}

\end{document}